\documentclass[a4paper,UKenglish,cleveref, autoref]{lipics-v2019} %

 \hideLIPIcs \nolinenumbers

\usepackage{todonotes}
\usepackage{xspace}
\usepackage{multicol}

\usepackage{amssymb}
\usepackage{bm}
\usepackage{mathtools}
\usepackage{stmaryrd}
\usepackage{mathpartir}

\usepackage{etoolbox}
\usepackage{collect}
\usepackage{environ}

\definecollection{proofs}
\definecollection{remarks}
\newcommand{\moveproof}[2]{\begin{collect}{proofs}{}{} \paragraph*{Proof of \Cref{#2}:}\begin{proof}#1\end{proof}\end{collect}}
\newcommand{\moveappendix}[2]{\begin{collect}{remarks}{}{}#2 #1\end{collect}}

\ifcsdef{internaldetails}{
  \NewEnviron{onlyfullversioncomment}[1]{\BODY}
  \NewEnviron{onlyinternaldetails}[1]{\BODY}
  \NewEnviron{onlyfullversion}[1]{\BODY}
  \NewEnviron{onlyicalpversion}{}
  \NewEnviron{mproof}[1]{\begin{proof}\BODY\end{proof}}
}{}
\ifcsdef{fullversion}{
  \NewEnviron{onlyfullversioncomment}[1]{\BODY}
  \NewEnviron{onlyinternaldetails}[1]{}
  \NewEnviron{onlyfullversion}[1]{\BODY}
  \NewEnviron{onlyicalpversion}{}
  \NewEnviron{mproof}[1]{\expandafter\moveproof\expandafter{\BODY}{#1}}
}{}

\ifcsdef{icalpversion}{
  \NewEnviron{onlyfullversioncomment}[1]{}
  \NewEnviron{onlyinternaldetails}[1]{}
  \NewEnviron{onlyfullversion}[1]{\expandafter\moveappendix\expandafter{\BODY}{#1}}
  \NewEnviron{onlyicalpversion}{\BODY}
  \NewEnviron{mproof}[1]{\ifcsdef{inappendix}{
      \begin{proof}
        \BODY
      \end{proof}}{\expandafter\moveproof\expandafter{\BODY}{#1}}}
}
{}

\NewEnviron{technicalappendix}[1]{\expandafter\moveappendix\expandafter{\BODY}{#1}}

\usepackage{coqtheorem}
\setBaseUrl{{https://ps.uni-saarland.de/extras/wcbv-reasonable/doc/}}

\usepackage[style=english]{csquotes}
\let\doquote\enquote
\renewcommand\enquote[1]{\textit{\doquote{#1}}}

\usepackage{enumitem} %

\usepackage{macros}

\renewcommand{\L}{\ensuremath{\mathsf{L}}\xspace}
\renewcommand{\bar}[1]{\overline{#1}}

\newcommand{\tenc}[1]{\raisebox{-1pt}{\ensuremath{{}^\ulcorner}}{\hspace{-3pt}#1} \hspace{-2pt}\raisebox{-1pt}{\ensuremath{{}^\urcorner}}}
\newcommand{\encode}[1]{\tenc{#1}}%

\title{The Weak Call-By-Value $\lambda$-Calculus is Reasonable for Both Time and Space}

\titlerunning{The Weak Call-By-Value $\bm\lambda$-Calculus is Reasonable for Both Time and Space}

\author{Yannick Forster}{Saarland University, Saarland Informatics Campus (SIC), Saarbrücken, Germany}{forster@ps.uni-saarland.de}{}{}%

\author{Fabian Kunze}{Saarland University, Saarland Informatics Campus (SIC), Saarbrücken, Germany}{kunze@ps.uni-saarland.de}{}{}%

\author{Marc Roth}{Cluster of Excellence (MMCI), Saarland Informatics Campus (SIC), Saarbrücken, Germany}{mroth@mmci.uni-saarland.de}{https://orcid.org/0000-0003-3159-9418}{}

\authorrunning{Y. Forster and F. Kunze and M. Roth}

\Copyright{Yannick Forster and Fabian Kunze and Marc Roth}%



\supplement{A \begin{onlyicalpversion}{}full version of the paper and a\end{onlyicalpversion} Coq formalisation of the results in Section~\ref{sec:prelims} and Section~\ref{sec:abstract_machines} \begin{onlyfullversioncomment}{}is\end{onlyfullversioncomment}{}\begin{onlyicalpversion}{}are\end{onlyicalpversion}{} available at \url{https://ps.uni-saarland.de/extras/wcbv-reasonable}.}%

\EventEditors{John Q. Open and Joan R. Access}
\EventNoEds{2}
\EventLongTitle{42nd Conference on Very Important Topics (CVIT 2016)}
\EventShortTitle{CVIT 2016}
\EventAcronym{CVIT}
\EventYear{2016}
\EventDate{December 24--27, 2016}
\EventLocation{Little Whinging, United Kingdom}
\EventLogo{}
\SeriesVolume{42}
\ArticleNo{23}

\newcommand{\FK}[2][]{\todo[color=orange, #1]{Fabian: #2}}

\newcommand{\True}{\mathtt{true}}
\newcommand{\False}{\mathtt{false}}

\begin{document}

\maketitle

\begin{abstract}
We study the weak call-by-value $\lambda$-calculus as a model for computational complexity theory and establish the natural measures for time and space -- the number of beta-reductions and the size of the largest term in a computation -- as reasonable measures with respect to the invariance thesis of Slot and van Emde Boas [STOC~84]. More precisely, we show that, using those measures, Turing machines and the weak call-by-value $\lambda$-calculus can simulate each other within a polynomial overhead in time and a constant factor overhead in space for all computations that terminate in (encodings) of \enquote{true} or \enquote{false}. We consider this result as a solution to the long-standing open problem, explicitly posed by Accattoli [ENTCS~18], of whether the natural measures for time and space of the $\lambda$-calculus are reasonable, at least in case of weak call-by-value evaluation.

Our proof relies on a hybrid of two simulation strategies of reductions in the weak call-by-value $\lambda$-calculus by Turing machines, both of which are insufficient if taken alone. The first strategy is the most naive one in the sense that a reduction sequence is simulated precisely as given by the reduction rules; in particular, all substitutions are executed immediately. This simulation runs within a constant overhead in space, but the overhead in time might be exponential. The second strategy is heap-based and relies on structure sharing, similar to existing compilers of eager functional languages. This strategy only has a polynomial overhead in time, but the space consumption might require an additional factor of $\log n$, which is essentially due to the size of the pointers required for this strategy. Our main contribution is the construction and verification of a space-aware interleaving of the two strategies, which is shown to yield both a constant overhead in space \emph{and} a polynomial overhead in time.
\end{abstract}

\section{Introduction}

Turing machines are the de-facto foundation of modern computability and complexity theory, in part due to the conceptual simplicity of their definition.
However, this simplicity is also one of the biggest disadvantages:
When it comes to detailed or formal reasoning, Turing machines soon become impossible to treat, because they lack compositionality and heavy logical machinery has to be used to reason about them.
This is best reflected by the fact that modern day researchers in computability and complexity theory usually have not faced explicit Turing machines since their undergraduate studies.
Instead, it is common to rely on pseudo code or mere algorithmic descriptions.
For computability theory, other models of computation like RAM machines, recursive functions or variants of the $\lambda$-calculus can be used if details are of interest, because the notion of computation is invariant under changing the model.
Especially the $\lambda$-calculus shines in this aspect, because tree-like inductive datatypes can be directly encoded and equational reasoning is accessible to verify the correctness of programs, which even makes the $\lambda$-calculus feasible as a model to formalise computability theory in proof assistants~\cite{Norrish2011,Forster17}.
However, this notion of \emph{invariance} does not suffice for complexity theory.
As stated by Slot and van Emde Boas~\cite{Slot:1984:TVC:800057.808705}:
\begin{center}
\textit{\enquote{Reasonable} machines can simulate each other within a polynomially bounded overhead in time and a constant factor overhead in space.}
\end{center}
If only reasonable machines are considered, this invariance thesis makes complexity classes robust under changing the model of computation. %
Until now, only sequential models of computation have been shown to fulfil this strong notion of invariance with natural complexity measures for time and space~\cite{dershowitz_invariance_nodate}.
The time and space complexity measures known to be reasonable for the full $\lambda$-calculus are ``total ink used'' and ``maximum ink used'' in a computation~\cite{lawall1996optimality}.
While the notion for space is natural, the notion for times is very unnatural and of no real interest.
Other measures rely on concrete implementations, giving no satisfying answer to the question whether the $\lambda$-calculus can be considered reasonable.

Dal Lago and Martini~\cite{DalLagoMartini08} gave a preliminary result in 2008 for the weak call-by-value $\lambda$-calculus and showed that counting $\beta$-steps while taking the size of $\beta$-redices into account is a reasonable measure for time.
In 2014 Accattoli and Dal Lago~\cite{indeed} showed that counting (leftmost-outermost) $\beta$-steps makes the full $\lambda$-calculus reasonable for time, starting a long line of research regarding measures for and implementations of the $\lambda$-calculus (see e.g.~\cite{accattoli2017complexity}).
Whether the natural measure for space, i.e. the size of the largest term in a computation, can be used together with the number of $\beta$-steps or how it has to be adapted is a long-standing open problem.

We solve this problem for the deterministic weak call-by-value $\lambda$-calculus (which we call $\L$) and show that the size of the largest intermediate term in a reduction makes \L a reasonable machine in the strong sense.
We consider our solution more than just a partial solution on the way to answering the question for the full $\lambda$-calculus in several aspects:
First, weak call-by-value evaluation is the standard model of eager functional programming languages.
Second, \L is already Turing-complete and one does not gain more power for the implementation or verification of algorithms by the strong reduction allowed in the full $\lambda$-calculus.
Third, from the complexity-theoretic point of view, the problem is solved: A certain form of the $\lambda$-calculus can be used to spell out arguments.
However, from an implementation point of view, many questions remain open: Our simulation uses the freedom given by ``polynomial overhead'' and should not be seen as a proposal for a uniform, canonical implementation, which is still very much desirable.

In what follows, we explain how to simulate \L on Turing machines with polynomial overhead in time and linear overhead in space, based on the natural measures, and vice-versa.

\begin{definition}\label{def:measures}
For a closed term $s$ that reduces to a normal form $\lambda x. u$
\[s = s_0 \red s_1 \red \cdots \red s_k=\lambda x.u \]
we define the time consumption of the computation to be $\Time{s}=k$ and the space consumption to be $\Space{s}= \max_{i=0}^k \size{s_i} $ where $\size s$ is the size of $s$.%
\end{definition}

\subparagraph{Result and Discussion}
We prove that the weak call-by-value $\lambda$-calculus is a reasonable machine with respect to the natural time and space measures defined above. For the formal statement we fix a finite alphabet $\Sigma$ and say that a function $f:\Sigma^\ast \to \Sigma^\ast$ is \emph{computable} by $\L$ in time $\mathcal{T}$ and space $\mathcal{S}$ if there exists an $\L$-term $s_f$ such that for all $x\in \Sigma^\ast$ we have that
\[s_f \encode{x} \red^\ast \encode{f(x)}~~\text{and}~~ \Time{s_f \encode x} \leq \mathcal{T}(|x|) ~~\text{and} ~~ \Space{s_f \encode x} \leq \mathcal{S}(|x|) \,.\]
Here $\encode{\,\cdot\,}$ is an encoding of strings over $\Sigma$.
\begin{theorem}[]\label{thm:intro_main}
Let $\Sigma$ be a finite alphabet such that $\{\True,\False\}\subseteq \Sigma$ and let $f:\Sigma^\ast \rightarrow \{\True,\False\}$ be a function. Furthermore, let $\mathcal{T}, \mathcal{S} \in \Omega(n)$.
\begin{enumerate}
\item If $f$ is $\L$-computable in time $\mathcal{T}$ and space $ \mathcal{S}$, then $f$ is computable by a Turing machine in time $\bigO{\poly{\mathcal{T}(n)}}$ and space $\bigO{ \mathcal{S}(n)}$.
\item If $f$ is computable by a Turing machine in time $\mathcal{T}$ and space $ \mathcal{S}$, then $f$ is $\L$-computable in time $\bigO{\poly{\mathcal{T}(n)}}$ and space $\bigO{ \mathcal{S}(n)}$.
\end{enumerate}
\end{theorem}
The conditions $\mathcal T\in \Omega(n)$ and $\mathcal S \in \Omega(n)$ state that we do not consider sublinear time and space. Furthermore, the restriction of $f$ to $\{\True,\False\}$ can be seen as a restriction to characteristic functions, which is sufficient for the complexity theory of decidability problems.

To the best of our knowledge this is the first proof of the invariance thesis including time \textit{and space} for a fragment of the $\lambda$-calculus using the natural complexity measures. %

\newcommand\expl{s_E}
At this point the reader might have the following objection:
A well-known problem in the $\lambda$-calculus is the issue of size explosion.
There exist terms that reduce to a term of size $\Omega(2^n)$ with only $\bigO{n}$ beta reductions.
Let us adapt an example from~\cite{DalLagoMartini08}:
Given a natural number $n$, we write $\overline n$ for its Church-encoding, defined as 
\[\overline n := \lambda f x.(\underbrace{f(f(f\cdots (f}_{n \text{ times}} x)\cdots)))\]
and define the Church encoding of the Boolean $\True$ as $\overline{\True} := \lambda xy.x$.
Next we define the term $\expl := \lambda x. \overline \True ~ \overline \True ~ (x\overline{2} (\lambda x. x))$.
Note that the application $\overline n \, \overline 2$ reduces to a normal form~$t_n$ of size $\Omega(2^n)$, extensionally equivalent to the Church-exponentiation $\overline{2^n}$.
The term~$\expl$ thus encodes a function, computing an exponentially big intermediate result, discarding it and returning $\overline \True$.
Formally, we have
\[\expl ~ \bar{n} %
  \red^4 (\lambda y.\overline \True) ~ (\underbrace{\overline 2 (\overline 2 \dots (\overline 2}_{n \text{ times}} (\lambda x.x)))) \underbrace{\red \dots \red}_{\bigO{n}\text{ times}} (\lambda y. \overline \True) ~ t_n \red \overline{\True} \]
for a term $t_n$ with $\size{t_n} \in \Omega(2^n)$.
Now $\Time{\expl ~ \bar{n}} \in \Theta(n)$, i.e., $\expl ~ \bar{n}$ reduces to $\bar{\True}$ in about $n$ beta reductions.
Moreover, $\Space{\expl ~ \bar{n}} \geq \size{(\lambda y.\overline \True)~t_n} \in \Omega(2^n)$, i.e., the largest term in the reduction is of exponential size.
While it might seem counterintuitive that a reasonable machine allows a computation that requires much more space than time, which is impossible for Turing machines, we consider this as one of the major insights of this work:
such computations returning values of bounded size can always be optimised, i.e. the size explosion is unnecessary to compute the result of the term.
Of course, any reduction sequence that \emph{ends} in a term of exponential size cannot be simulated in less space if the result has to be written down explicitly.
However, in complexity theory of decision problems, the functions that matter in the end are characteristic functions which map to $\True$ and $\False$.
We make this more precise:
For any size-exploding term $\expl$, we can use \Cref{thm:intro_main} twice to obtain a term $\widehat \expl$ computing the same function with polynomial space usage.
By \Cref{thm:intro_main} (1) there is a Turing machine that on an encoding of $\bar{n}$ simulates $\expl ~ \bar{n}$ with time and space complexity of $\bigO{n^c}$ (the latter since Turing machines can not use more space than time).
By \Cref{thm:intro_main} (2), there is a term $\widehat \expl$ s.t. $\widehat \expl ~ \overline n$ has the same normal form as $\expl ~ \overline n$ -- but with space complexity $\bigO{n^c}$, since the overhead in space is constant-factor.

In~\cite{accattoli2018efficiency}, Accattoli writes
\enquote{Essentially one is assuming that space in the $\lambda$-Calculus is given by the maximum size of terms during evaluation, and since in sequential models time is greater or equal to space (because one needs a unit of time to use a unit of space), the time cost of the size exploding family must be at least exponential. The wrong hypothesis
then is that space is the maximum size of the terms during evaluation. It is not
yet clear what accounts for space in the $\lambda$-Calculus, however the study of time cost models made clear that space is not the size of the term.}
Our result implies that this conclusion does not apply for weak call-by-value evaluation in $\L$.
In this particular case, the wrong hypothesis is that even when simulating on Turing machines, where time is greater or equal to space, the space-measure of a $\lambda$-calculus term does not have to coincide with the actual resources used when simulating it and can thus be much larger than the time-measure. Furthermore, we point out that there exist reasonable sequential models that might consume asympotically more space than time as illustrated in Appendix~\ref{sec:sveb} in case of RAM machines.

\subparagraph*{Simulation Strategies}
In the previous paragraph we argued that our proposed cost and time measures for $\L$ are not inherently contradictory with respect to the invariance thesis. However, we did not provide explicit information of the simulations yet, which we are going to catch up on now.
Note that a simulation of Turing machines in $\L$ for the second part of Theorem~\ref{thm:intro_main} regarding time has already been given by Accattoli and Dal Lago~\cite{ADL,lago_encoding_2017}.
We argue in \Cref{sec:TM_in_L} that their construction also only has a constant factor overhead in space and thus works for our purposes as well.
The main part of the paper thus focuses on the simulation of $\L$ by Turing machines.
Essentially, we rely on an interleaving of two different strategies for simulating a reduction in $\L$ with a Turing machine.
Both strategies are formally introduced in Section~\ref{sec:abstract_machines}; in what follows, we provide an intuitive overview.

The first one, which we call the \emph{substitution-based strategy} %
simulates a reduction sequence naively as given by the reduction rules of $\L$.
In particular, all substitutions are executed immediately if a $\beta$-reduction is performed.
However, we have already seen an example which shows this strategy to be insufficient: Consider again the term $\expl~\bar{n}$ which reduces to $\bar{\True}$ in $\Theta(n)$ beta reductions. If this reduction sequence is simulated naively, exponentially many substitutions have to be performed, and hence the time consumption of that machine would be exponential in $n$.
At the same time, for this term, the strategy is valid if we would only care for space, because the space complexity of the \L term is already exponential.
In general, we show that any reduction sequence in $\L$ can be simulated with only a constant overhead in space by a Turing machine using the substitution-based strategy.

Solving the issue regarding the time consumption requires us to rely on the second simulation strategy which we call the \emph{heap-based strategy}. Intuitively, we do not execute any substitution if a $\beta$-reduction is simulated. Instead we use \emph{closures} and keep track of the values assigned to variables in an environment. These environments are stored on an explicit heap containing pointers and terms. This allows for \emph{structure sharing}, similar to real-world execution of functional languages as well as to the strategy used in~\cite{indeed}.
Indeed, applying this strategy to the reduction sequence $\expl~\overline{n} \red \dots \red \bar{\True}$ yields a polynomial number of steps in a simulation with a Turing machine.

At this point, one might be tempted to think that the heap-based strategy is strictly superior to the substitution-based strategy. However, there is one (major) catch: There exist reduction sequences of time and space linear in the input term size $n$, which yield an overhead of factor $\log n$ in space when simulated using the heap-based strategy. The reason is that the number of heap entries is linear, which requires the pointers, i.e.\ the heap addresses, to grow in size. 
The following example illustrates this phenomenon:
Let $\text{N} := (\lambda{xy}.{xx})\overline\True$.
\newcommand\pexp{s_P}
\begin{align*}
\pexp := \underbrace{\text{N}(\cdots (\text{N}}_{n \text{ times}} \overline\True)\ldots) \succ^{n}\underbrace{{(\lambda{y}.{\overline\True \, \overline\True)}}(\cdots ((\lambda{y}.{\overline\True \,  \overline\True})}_{n \text{ times}}  \overline\True)\ldots) \succ^{2n} \overline\True 
\end{align*}
Since $\pexp$ performs $3n$ beta reductions it needs $3n$ entries on the heap.
The heap pointers then make the space consumption ``explode'' again.
They are of size $\log n$ if binary numbers are used and $n$ if unary numbers are used, resulting in an overall space consumption of $\Omega(n \log n)$ or $\Omega(n^2)$, both forming more than constant factor overhead.
We call this problem \emph{pointer explosion}, analogous to the discussed size explosion problem, and point out that both phenomena have already been identified and discussed by Slot and van Emde Boas~\cite{Slot:1984:TVC:800057.808705} in their treatment of RAM simulations by Turing machines.

In our case of the weak call-by-value $\lambda$-calculus $\L$, we have obtained two simulation strategies, each solving \emph{one} of the problems: The substitution-based strategy works for space, but is insufficient for time on terms exhibiting size explosion (i.e.\ which have exponentially big intermediate terms).
The heap-based strategy works for time, but is insufficient for space on terms exhibiting pointer explosion.
The crucial observation is now that on terms exhibiting size explosion, i.e.\ reaching a term of size $\Omega(2^n)$ in $n$ steps, pointer explosion is a non-issue:
$n$ pointers of size $\log n$ can easily be accommodated for in space~$\bigO{2^n}$.

Since it is a-priori not decidable whether a term exhibits size explosion or pointer explosion, we interleave the execution of the two simulation strategies.
We simulate the execution for every step number $k$ repeatedly, and always try to run the subsitution-based strategy first.
If the size of intermediate terms becomes big enough to accommodate exploding pointers, we immediately abort and try the heap-based strategy for $k$ steps instead.
The heap-based strategy is thus guaranteed to not encounter the pointer explosion problem and the substitution-based strategy can not encounter the size explosion problem, because it is aborted beforehand.
The details of this interleaving machine, which we consider as our main technical contribution, are given in~\Cref{sec:L_in_TM}.

\subparagraph{Formalisation in Coq}

A technically demanding and error-prone part of our proof is analysing the exact complexity of the abstract machines involved (see \Cref{sec:abstract_machines}).

Because our proofs rely on many simulation notions containing hard-to-check side conditions, we provide a formalisation of all results for the abstract machines, i.e.\ every theorem and definition needed for and including \Cref{sec:abstract_machines}, in the proof assistant Coq~\cite{Coq}.\!\footnote{The code is hyperlinked with the PDF version of this document and can be accessed at \url{https://ps.uni-saarland.de/extras/wcbv-reasonable} or by clicking on the formalised statements and definitions, which are marked with a \includegraphics[height=1em]{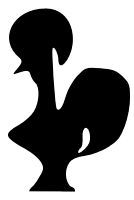}-symbol.}

A formalisation of the full results, including a formal verification of the Turing machines involved, is an ongoing and challenging project.
We reported on the project previously in~\cite{LOLA}; the current paper presents the finalised theoretical contributions.

\section{Preliminaries}\label{sec:prelims}

We adopt a notation closely related to type-theory, but the paper can be read with no background in type theory.
All defined functions are always total.
We use the type $\opt{X}$ to denote the type $X$ enriched with a new element $\bot$.
This allows us to view $X \to \opt{Y}$ as the type of partial functions from $X$ to $Y$.
In the definition of such partial functions, left out cases are meant to default to $\bot$.
Concerning lists $A,B:\List{X}$ over $X$, we use $[]$ for the empty list, write $x::A$ to prepend an element to a list and write $[x_1,\ldots,x_k]$ for a list built from the elements $x_i$. We write $A\con B :\List{X}$ for list concatenation%
, $\length{A}:\nat$ for length
and $A[n]:\opt{X}$ for list lookup.

\subsection{Call-by-value $\lambda$-calculus \L}
\setCoqFilename{LM.L}

The call-by-value $\lambda$-calculus introduced by Plotkin~\cite{Plotkin75} in his seminal paper is known to be a reasonable machine for time complexity~\cite{DalLagoMartini08}.
In those works, abstractions and variables are treated as values, but $\beta$-reduction below binders is not allowed, i.e. reduction is \textit{weak}.
We use a deterministic version of the weak call-by-value $\lambda$-calculus we call \L, originally introduced in~\cite{Forster17}.
We treat only abstractions as values but keep the weak behaviour of reduction.
On closed terms, the number of steps to a normal form agrees with the number of steps needed in the version in~\cite{DalLagoMartini08}.
We keep the definitions short and use the same notations as in~\cite{cbvlcm2}, where more details can be found.

We define the syntax of the $\lambda$-calculus using a de Bruijn representation of terms~\cite{deBruijn1972}: $s,t,u,v ~:~\Ter~::=~ n\mid st\mid\lambda s$ where $n:\nat$.
\begin{definition}[][subst]
We define a recursive function \emph{$\subst sku$}
providing a single-point, capturing \emph{substitution} operation:
\begin{align*}
  \subst kku
  &~:=~u
  &
  \subst nku
  &~:=~n
  \quad\quad\quad\text{if}~n\neq k
  \\
  \subst{(st)}ku
  &~:=~(\subst sku)(\subst tku)
  &
  \subst{(\lambda s)}ku
  &~:=~\lambda(\subst s{\natS k}u)
\end{align*}
\end{definition}
We say that $s$ is bounded by $k$ if all free de Bruijn indices in $s$ are lower than $k$.
Consequently, $s$ is a closed term iff it is bounded by $0$.

\begin{definition}[][step]
  We define a deterministic inductive 
\emph{reduction relation $s\red t$}, which is weak, call-by-value and agreeing with the reduction relation in~\cite{Plotkin75, DalLagoMartini08} on closed terms:
\begin{mathpar}
  \inferrule*{~}{(\lambda s) (\lambda t)\red\subst s0{\lambda t}}
  \and
  \inferrule*{s\red s'} {st\red s't}
  \and
  \inferrule* {t\red t'} {(\lambda s)t\red(\lambda s)t'}
\end{mathpar}
\end{definition}

Note that the only closed, irreducible terms are abstractions.
We write $\TimeBS{s}{k}{t}$ and $\SpaceBS{s}{m}{t}$ if $s \red^k t$ for $t$ being an abstraction and $m = \Space{s}$ as defined in \Cref{def:measures}.

The size of a term is defined with a unary encoding of indices in mind:%
\begin{align*}
  \size{n}&~:=~1+n &
  \size{\lambda s}&~:=~1 + \size{s} &
  \size{st}&~:=~1 + \size{s} + \size{t}
\end{align*}
For a binary encoding, i.e.\,$\size{n}~:=~1+\log_2 n$, we conjecture that the remainder of this paper can be adapted with no essential change.

\subsection{Encoding Terms as Programs}
\label{sec:programs}
\setCoqFilename{LM.Programs}

Turing machines can not directly operate on tree-like data structures like $\L$-terms.
We encode terms as \emph{programs} $P,Q,R
~:~\N{\Pro}$, which are lists of \emph{commands}:
\[c~:~\N{\Com} ~::=~  \ret\mid \var\,n \mid  \lamb \mid  \app \hspace{3em} \text{ with }n:\nat\] 
  \begin{definition}[][compile]\label{gamma}
    The encoding function $\gamma : \Ter \to \Pro$ compiles terms to programs:
\begin{mathpar}
  \gamma n~:=~[\var\,n] \and
  \gamma(st)~:=~\gamma s \con \gamma t \con [\app]  \and
  \gamma(\lambda s)~:=~\lamb :: \gamma s \con [\ret]
\end{mathpar}
\end{definition}
This encoding is similar to postfix notation, the additional command $\lamb$ makes it easier to detect subprograms representing values when traversing the encoding.

\begin{definition}[][reprP]
  We write $P \gg s$, read as \emph{$P$ represents $s$}, to connect programs with values in~\L.
  This relation is defined with the single rule $\inferrule*{}{\gamma t \gg \lambda t}$.
\end{definition}

\begin{technicalappendix}{}
  \setCoqFilename{LM.Programs}
\begin{lemma}[][compile_inj]\label{gamma-inj}
  $\gamma$ is injective.
\end{lemma}
\begin{mproof}{gamma-inj}
  We define an inverse $\delta: \nat \to \Pro \to \List \Ter \to \opt{\List \Ter}$ of $\gamma$:
  \begin{align*}
    \delta k (\var\,n::P) A&~:=~\delta k P (n::P) &
    \delta k (\app::P) (t::s::A)&~:=~\delta k P (st::A)\\
    \delta k (\lamb::P) A&~:=~\delta (\natS k) P A&
    \delta (\natS k) (\ret::P) (s::A)&~:=~\delta k P (\lam s::A)\\
    \delta k [] A&~:=~A &
  \end{align*}
  Now $\delta k (\gamma s\con P) A=\delta k P (s::A)$ holds by induction on $s$.
\end{mproof}
\end{technicalappendix}

To store the encoding of de Bruijn indices on tapes, we will use a unary encoding, motivating the following definition of the size of commands and programs:
\begin{mathpar}
  \size{\var\,n}~:=~1 + n \and
  \size{c}~:=~1\quad \text{if $c$ is no variable} \and
  \size{P}~:=~1 + \sum_{c \in P} \size{c}
\end{mathpar}
This size is compatible with term size, with factor $2$ due to the two commands for abstractions:
\begin{lemma}[][size_geq_1] \label{size-gamma}
  $1 \leq \size{s} \leq \size{\gamma s} \leq 2\size{s} - 1$
\end{lemma}
\begin{mproof}{size-gamma}
  Induction on $s$.
\end{mproof}

The use of $\lamb$ and $\ret$ to encode abstraction allows to define a function $\phi P: \opt{(\Pro \times \Pro)}$ that extracts the body of an abstraction by matching $\lamb$ with $\ret$ like parentheses. It uses an auxiliary function $\jumpTarget{k}{Q}{P}$ that stores the number $k$ of currently unmatched $\lamb$ and the prefix $Q$ already processed.
\begin{definition}[][jumpTarget]
  $\phi P ~:=~\jumpTarget{P}{0}{[]}$ with
  \begin{align*}    
  \jumpTarget 0  Q{(\ret::P)}&~:=~\some(Q,P) &
  \jumpTarget{\natS k} Q{(\ret::P)}&~:=~\jumpTarget k {Q \con[\ret]}P\\
  \jumpTarget k Q {(\lamb::P)}&~:=~\jumpTarget {\natS k} {Q \con[\lamb]}P &
  \jumpTarget k Q {(c::P)}&~:=~\jumpTarget k {Q\con[c]} P ~~\text{\normalfont if $c=\var\, n$ or $\app$}
\end{align*}
\end{definition}

\begin{lemma}[][jumpTarget_correct]\label{phi-gamma} %
  $\phi(\gamma s \con \ret :: P) = \some (\gamma s,P)$
\end{lemma}
\begin{mproof}{phi-gamma}
The generalisation $\phi_{k,Q}{(\gamma s \con P)} = \phi_{k,Q\con \gamma s} P$ follows by induction on $s$.
\end{mproof}

We define a substitution operation $\subst P k Q$ on programs, analogous to substitution on terms%
\begin{onlyicalpversion}%
  in \Cref{substP}%
\end{onlyicalpversion}.
\begin{onlyfullversion}{}
  \begin{definition}[Substitution on programs][substP]\label{substP} 
\begin{align*}
  \subst{(\var\,k::P)} k Q &~:=~  Q::\subst P k Q & 
  \subst{(\var\,n::P)} k Q &~:=~ \var\,n::\subst P k Q \\
  \subst{(\lamb::P)} k Q &~:=~ \lamb::\subst P {S k} Q&
  \subst{(\app::P)} k Q &~:=~ \app::\subst P k Q \\
  \subst{(\ret::P)} 0 Q &~:=~ [\ret] &
  \subst{(\ret::P)} {S k} Q &~:=~ \ret::\subst P k Q \\
  \subst{[]} k Q &~:=~ []
\end{align*}
\end{definition}
\end{onlyfullversion}
 Term and program substitution are compatible:
\begin{lemma}[][substP_correct]\label{subst-gamma} %
  $\subst {(\gamma s)}{0}{\gamma t} = \gamma (\subst s0 t)$
\end{lemma}
\begin{mproof}{subst-gamma}
  The generalisation $\subst {(\gamma s \con P)}{k}{\gamma t} = \gamma (\subst skt) \con \subst P k {\gamma t}$ holds by induction on s.
\end{mproof}

\subsection{Closures and Heaps}\label{sec:closures-heap}

\setCoqFilename{LM.AbstractHeapMachine}

To allow for structure sharing later, we introduce closures whose environments are stored in an explicitly modelled heap.
Environments are stored as linked lists of closures on the heap, and closures $g~:~\N{\HC} ~:=~ \Pro \times \HA$ contain programs and pointers $a,b~:~\N{\HA} ~:=~ \nat$ to the environment. We represent the heap $H~:~\N{\Heap}~:=~\List{\HE}$ as a list of its cells $e~:~\N{\HE}~::=~\envCons{g}{a}$ that store the head and the address of the tail.
To interpret the linked list structure, we define a \emph{lookup function} $\N{H[a,n]}:\opt{\HC}$ that returns the $n$-th entry of the list at address $a$, i.e. the value bound to the de Bruijn index $n$ in the environment $a$:
\begin{align*}
  H[a,n]&~:=~\ITE{n=0}{\some{g}}{H[b,n-1]} &&\text{where}~H[a]=\some{\envCons g b}
\end{align*}

The operation $\N{\M{put}~H\,e}:\Heap\times\HA$ puts a heap entry on the heap and returns the new heap and the address of the new element.
\begin{align*}
  \M{put}~H\,e&~:=~(H\con[e],~\natS\,|H|)
\end{align*}
In our setting, we allocate at the end of the heap and have no need for garbage collection.

A closure $(P,a)$ represents some term if the environment $a$ contains bindings for all free variables in $P$. To make this more precise, we first define the unfolding of a term relative to some environment:
\begin{definition}[Unfolding][unfolds] \label{unfolding}
  The unfolding $\inst{H}{k}{s}{a}{s'}$ is inductively given by the rules
\begin{mathpar}
  \inferrule*{n < k}
  {\inst{H}{k} n a n}
 \and
  \inferrule*{n \geq k
    \and H[a,n-k] = \some{(P,b)}
    \and P \gg s
    \and \inst H 0 s b {s'}}
  {\inst{H}{k} n a {s'}}
  \and
  \inferrule*{ \inst{H}{\natS k} s a {s'}}
  { \inst{H}{k} {\lambda s} a {\lambda s'}}
  \and
  \inferrule*{\inst{H}{k} s a {s'}
    \and \inst{H}{k} t a {t'}}
  {\inst{H}{k} {s t} {a}{s' t'}}
\end{mathpar}
\end{definition}
Intuitively, $\inst{H}{k}s a {s'}$ holds if $s'$ is obtained by recursively substituting all free variables in $s$ by their values in the environment $a$. The index $k$ is an artefact of the de Bruijn representation and denotes which variables in $s$ are locally bound during the traversal of $s$.

The first rule states that bound variables are left unchanged. The second rule states that for free variables, the environment $a$ binds $n$ to some value $s'$ that can be looked up in $H$.  The third rule descends under an abstraction and thus one more variable is considered bound in $s$. The last rule descends under application.

\begin{definition}[][representsCl]\label{rep-rel-clos}
  The relation $g \gg_H s$, read as \emph{$g$ represents $s$ relative to $H$}, is defined by the single rule
\begin{mathpar}
  \inferrule*{P \gg t \and \inst{H}{0} t a s}{(P,a) \gg_H s}
\end{mathpar}
\end{definition}
\begin{technicalappendix}{}
Adding a value $t'$ to an environment $a$ results in substitution in the unfolded term:
\setCoqFilename{LM.AbstractHeapMachine}
\begin{lemma}[][unfolds_subst]\label{inst-subst}
  If $H[a'] = \envCons g a$ with $g \gg_H t'$ and $\inst{H}{1} s a {s'}$, then $\inst{H}{0} s {a'} {\subst{s'}{0}{t'}}$.
\end{lemma}
\begin{mproof}{inst-subst}
  Let $H[a'] = \envCons g a$ and $g \gg_H t'$. We show a generalisation: \emph{If $\inst{H}{\natS k} s a {s'}$, then $\inst{H}{k} s {a'}{\subst{s'}{k}{t'}}$}, by induction on $\inst{H}{\natS k} s a {s'}$.
  
  In the case $s = n < \natS k$ and $s'=n$, there are two subcases: Assuming $n<k$,  $\inst{H}{k} n {a'} n = \subst{s'}{k}{t'}$ holds by definition. Otherwise, we have $n=k$. Since $H[a']= \envCons g a$, we have $H[a',n-k]=H[a',0]=g$. With $P,b$ such that $g=(P,b)$, we have $\inst{H}{k} n {a'}{\subst{n}{k}{t'}}=t'$ by the second rule since $(P,b) \gg_H t'$ implies $(P \gg t)$ and $ (t,b) \gg_H t'$ for some $t$.
  
  In the case $s = n \geq \natS k$, we have $H[a,n-(\natS k)]=(P,b)$ and $P \gg u$ with $\inst{H}{0} u b {s'}$ for some $P,b$. As $H[a'] = \envCons g a$, we have $H[a',n-k]=H[a,n-(\natS k)]=(P,b)$. Therefore $\inst{H}{k} n a {\subst {s'} {k} {t'}} = s'$ by the second rule, where the equality holds as $s'$ is closed by \Cref{inst-bound}.
  
  In the other cases, i.e. application and abstraction, the claim follows by the inductive hypothesis and the definition of $\inst{H}{k} \cdot \cdot \cdot $. 
\end{mproof}
\end{technicalappendix}

\begin{technicalappendix}{}
Unfolding only changes de Bruijn indices starting at $k$:
\begin{lemma}[][bound_unfolds_id]\label{bound-inst}
  If $s$ is bounded by $k$, then $\inst{H}{k}sas$.
\end{lemma}
\begin{proof}
  Induction on $s < k$.
\end{proof}
In particular, closed terms are invariant under unfolding.

\end{technicalappendix}

\begin{technicalappendix}{}%
The unfolding relation only holds if all de Bruijn indices up to $k$ are bound in $a$. 
\begin{lemma}[][unfolds_bound]\label{inst-bound}
  If $\inst{H}{k}sa{s'}$, then $s' < k$
\end{lemma}
\begin{proof}
  Induction on $\inst{H}{k}sa{s'}$.
\end{proof}
In particular for $k=0$, unfolding results in closed terms.

\end{technicalappendix}

\begin{technicalappendix}{}
  A heap is extended by another heap if the latter contains a superset of the entries:
  \begin{definition}[][extended]
   $H \subseteq H' := \forall a, H[a] \neq \none \rightarrow H[a] = H'[a]$
  \end{definition}
\begin{lemma}[][extended_PO]\label{extend-PO}
  Heap extension $H \subseteq H'$ is transitive and reflexive.
\end{lemma}

\begin{mproof}{extend-PO}
  Transitivity and reflexivity follow from the same properties for equality.
\end{mproof}
\end{technicalappendix}

\begin{technicalappendix}{}
Heap extension does not change the result of certain operations:
\begin{lemma}[][lookup_extend] \label{extend-compat}
  Assume $H \subseteq H'$
  \begin{enumerate}
  \item If $H[a,n] \neq \none$, then $H[a,n] = H'[a,n]$.
  \item If $\inst{H}{k} sa{s'}$, then  $\inst{H'}{k} sa{s'}$.
  \item If $g \gg_H s$, then  $g \gg_{H'} s$.
  \end{enumerate}
\end{lemma}

\begin{proof}
  The first claim follow by induction on $n$. The second claim follows by induction on $\inst{H}{k} s a {s'}$. The only interesting case is the one where $s=n \geq k$, which requires the first claim. The third claim follows from the second by definition of $\gg_H$.
\end{proof}
\end{technicalappendix}

\section{Abstract Machines}\label{sec:abstract_machines}

In order to analyse the two mentioned strategies on a more semantic level than just as implementations on Turing machines we introduce two abstract machines implementing these strategies  -- based on substitutions and based on heaps.
The machines are variants of the ones presented in $\cite{cbvlcm2}$.
Both machines will take $\bigO{\Time{s}}$ abstract steps to evaluate a term $s$, but differ in the size of intermediate states and in the complexity of their respective implementations as Turing machines, which we construct in \Cref{sec:L_in_TM}.

\subsection{Substitution Machine} \label{sec:subst-machine}
\setCoqFilename{LM.AbstractSubstMachine}

We define an abstract machine that uses substitution on programs. The implemented strategy is close to the small-step semantics for L.
One important property is that the size of machine states during the machine run is linear in the size of the intermediate terms.
Therefore, the substitution-based Turing machine will have constant factor overhead for space.

The abstract machine executes terms using two stacks of programs $T$ and $V$ called \emph{task} and \emph{value} stack. The task stack holds the parts of the program yet to be executed, and the value stack holds the already fully evaluated parts.

\begin{figure}
  \vspace{-1em}
  \small\begin{align*}
    (\lamb::P)::T,~V
    &~\red~\tailRec{P'}{T},~Q::V
    &&\text{if }\phi P= \some(Q,P')
    \\
    (\app::P)::T,Q::R::V
    &~\red~\subst R 0 {\lamb::Q\con[\ret]}::\tailRec{P}{T},~V
    && \\
       \text{where }   \tailRec{P}{T} &~:=~ \ITE{P=[]}{T}{P::T}
  \end{align*}
\vspace{-2em}
  \caption {Reduction rules of the substitution machine}
  \label{fig:subst-red}
\end{figure}

The semantics of the substitution machine is defined in \Cref{fig:subst-red}. The machine executes the first command of the topmost program of the task stack. In the \emph{lambda rule}, the command $\lamb$ marks the start of an abstraction. The sub-program corresponding to the body of the abstraction is moved to the value stack. In the \emph{application rule}, the topmost values are applied to each other: The program $R$ is instantiated with the argument $Q$ to obtain a new task to be evaluated.

We need tail call optimisation $\tailRec{}{}$ to guarantee that the size of the machine state is linear in the size of the represented term.
Without it, the application rule could pile up return-tasks $P=[]$ inside the task stack, invalidating \Cref{subst-space}. %
\begin{onlyfullversioncomment}{}
While we only need tail call optimisation for the application rule, adding it to the lambda rule as well streamlines proofs and allows us to avoid a rule to discard empty programs.
\end{onlyfullversioncomment}{}
The \emph{initial state $\tau_s$} for a term $s$ is $\tau_s~:=~([\gamma s],[])$.

The machine evaluates \L with a number of steps linear in the time-measure:

\begin{theorem}[Substitution machine runtime][correctTime']\label{subst-time} %
  If $\TimeBS{s}{k}{t}$, then
  $\tau_s \red^{3k+1} ([],[P])$ for some $P$ with $P \gg t$.
\end{theorem}
\begin{mproof}{subst-time}
  We show the generalisation, \emph{if $\TimeBS{s}{k}{t}$, then for all $Q,T,V$ we have $((\gamma s \con Q) :: T,V) \red^{3k+1} (\tailRec{Q}{T},P::V)$ for some $P$ with $P \gg t$}, from which the claim follows for $Q=R=V=[]$.

  Proof by induction on $\TimeBS{s}{k}{t}$ as defined in \Cref{def:time-bs}.  

  In the case $\TimeBS{\lambda s}{0}{\lambda s}$, we have 
\begin{align*}
  ((\gamma (\lambda s) \con Q) :: T,V) =& ((\lamb :: \gamma s :: \ret \con Q) :: T,V)&&\\
  \red& (\tailRec{Q}{T},\gamma s::V)&&\text{\Cref{phi-gamma}}
\end{align*}
and $\gamma s \gg s$ holds by definition.

In the case $\TimeBS{s t}{k_1 + k_2 + 1 + k_3}{u}$ with all names as in \Cref{def:time-bs}, we have 
\begin{align*}
  ((\gamma ( s t) \con Q) :: T,V) &= ((\gamma s \con \gamma t \con \app :: Q) :: T,V)&&\\
&\red^{3k_1 + 1} (\tailRec{(\gamma t \con \app :: Q)}{T},\gamma s'::V)&&\text{IH for $\TimeBS{s}{k_1}{\lambda s'}$}\\
&\red^{3k_2 + 1} (\tailRec{(\app :: Q)} {T},\gamma t'::\gamma s'::V)&&\text{IH for $\TimeBS{t}{k_2}{\lambda t'}$}\\
&\red ((\subst {(\gamma s')}{0}{\gamma (\lambda t')} :: \tailRec Q T,\gamma t'::\gamma s'::V)\\
&= (\gamma (\subst {s'}{0}{\lambda t'}) :: \tailRec Q T,V)&&\text{\Cref{subst-gamma}}\\
      &\red^{3k_3+1} (\tailRec Q T ,\gamma u::V)&&\text{IH for $\TimeBS{\subst{s'}{0}{\lambda t'}}{k_3}{\lambda u}$}
\end{align*}
Note that $\tailRec{}{}$ is $::$ in the first two reductions.

The claim follows as $3(k_1+k_2+1+k_3) + 1 = (3k_1 + 1) + (3k_2 + 1) + 1 + (3k_3 + 1)$ and $\gamma u \gg \lambda u$ by definition.
\end{mproof}

The size $\size T$ of $T$ is defined to be just the sum of the sizes of the elements in $T$, and similar for $V$. The size of a state is defined by $\size{(T,V)}=\size{T}+\size{V}$. We write $\tau \red^*_m \tau' $ for a sequence of machine reductions where the largest state has size $m$.

The maximal machine state size when evaluating $s$ is asymptotically as large as $\Space{s}$. 
\begin{theorem}[Substitution machine state size][correctSpace']\label{subst-space}  %
  If $\SpaceBS{s}{m}{t} $, then $\tau_s \red^*_{m'} ([],[P])$ for some $P$ and $m'$ with $P \gg t$ and $m \leq m' \leq 2m$.
\end{theorem}
\begin{mproof}{subst-space}
  We show a generalisation, \emph{if $\SpaceBS{s}{m}{t}$, then for all $Q,T,V$ we have $((\gamma s \con Q)::T,V) \red^*_{m'} (\tailRec Q T,P::V)$ for some $P$,$m'$ with $P \gg t$ and $ m+\size {\tailRec P T}+\size V \leq m' \leq 2m+\size {\tailRec P T}+\size V$}, from which the claim follows with $Q=T=V=[]$.
  
 Proof by induction on $\SpaceBS{s}{m}{t}$ as defined in \Cref{def:space-bs}. By definition of $\SpaceBS{}{}{}$, this proof is very similar to the one for \Cref{subst-time}. The only difference is the needed equalities between the various space-measures $m_i$. Those are proven by tedious, but straightforward computations when using the facts that $\size{P} + \size{T} \leq \size{\tailRec{P}{T}} \leq \size{P} + \size{T} + 1$ and \Cref{subst-gamma} and \Cref{size-gamma} and the fact that $\SpaceBS{s}{m}{t}$ implies $ \size{s} \leq m \geq \size{t}$.
\end{mproof}

\subsection{Heap Machine}\label{sec:heap-machine}
\setCoqFilename{LM.AbstractHeapMachine}
This machine uses the heap described in \Cref{sec:closures-heap} to enable sharing of environments. One important feature of this machine is that the size of intermediate states does only depend on the size of the initial term $s$ and the number of machine steps, but not on $\Space{s}$.

\begin{figure}
  \small\vspace{-1em}
  \begin{align*}
    (\var\,n :: P,a)::T,~V,~H
    &~\red~(P,a)::T,~g::V,~H
    &&\text{if }H[a,n]=\some{g}
    \\
    (\lamb :: P,a)::T,~V,~H
    &~\red~(P',a)::T,~(Q,a)::V,~H
    &&\text{if }\phi P= \some(Q,P')
    \\
    (\app::P,a)::T,~g::(Q,b)::V,~H
    &~\red~(Q,b')::(P,a)::T,~V,~H'
    &&\text{if }\M{put}\,H\,\envCons{g}{b}={(H',b')}
    \\
    ([],a)::T,~V,~H
    &~\red~T,~V,~H
    &&
  \end{align*}
  \vspace{-2em}
  \caption {Reduction rules of the heap machine}
  \label{fig:heap-red}
\end{figure}

The machine is defined in \Cref{fig:heap-red}. Its task and value stacks contain closures. The \emph{variable rule} loads the value bound to a variable to the value stack. The \emph{lambda rule} copies a subprogram representing an abstraction to the value stack. The \emph{application rule} calls the subprogram $Q$ after adding the value $g$ as argument to the environment of $Q$. The \emph{return rule} drops finished tasks. The use of closures instead of programs allows an explicit variable rule instead of program-level substitution.
The \emph{initial state $\sigma_s$} for a closed term $s$ is $([(\gamma s,0)],[],[])$ as $ \inst{H}{0}{s}{0}{s}$ by \Cref{bound-inst}.

The machine evaluates \L with a number of steps linear in the time-measure:
\begin{theorem}[Heap machine runtime][correctTime']\label{closTime} If $\TimeBS{s}{k}{t}$ and $s$ is closed, then
  $\sigma_s \red^{4k+2} ([],[g],H)$ for some $g,H$ with $g \gg_H t$.
\end{theorem}
\begin{mproof}{closTime}
  We show a generalisation, \emph{
  If $\TimeBS{s}{k}{t}$ and $\inst{H}{0}{s_0}as$, then there are $g$ and $H'$ with $g \gg_{H'} t$ such that $((\gamma s_0 \con P,a)::T,V,H) \red^{4k+1} ((P,a)::T,g::V,H')$ for any $P,T,V$, and $H \subseteq H'$}. Here $H \subseteq H'$ is meant as  in \Cref{coq:extended}. The original claim follows with $P=T=V=H=[]$, \Cref{bound-inst} and the reduction rule for empty tasks.

Proof by induction on $\TimeBS{s}{m}{t}$ as in \Cref{def:time-bs}.
In the case of $\TimeBS{\lambda s}{0}{\lambda s}$, a case distinction on $\inst{H}{0} {s_0} a {\lambda s}$ yields two cases: $s_0$ is either a variable with a value bound in $a$, or $s_0$ is an abstraction.

In the case $s_0=n$, we obtain $Q,b,s_1$ such that $H[a,n] = \some (Q,b)$ with $Q \gg s_1$ and $\inst{H}{0}{s_1} b {\lambda s}$. The claim holds as $(Q,b) \gg_H \lambda s$ and $((\gamma n \con P,a)::T,V,H) = ((\var\,n ::P,a)::T,V,H)
\red ((P,a)::T,(Q,b)::V,H)$.

In the case $s_0=\lambda s_1$, we have that $\inst{H}{1} {s_1} a s$. The claim holds as $(\gamma s_1,a) \gg_H \lambda s$ and $((\gamma (\lambda s_1) \con P,a)::T,V,H) = ((\lamb :: \gamma s_1 \con \ret :: P,a)::T,V,H)
\red ((P,a)::T,(\gamma s_1,a)::V,H)$.

In the other case of the induction, $\TimeBS{st}{k_1 + k_2 + 1 + k_3}{u}$, we have $\TimeBS{s}{k_1}{\lambda s'}$ and $\TimeBS{t}{k_2}{\lambda t'}$ and $\TimeBS{\subst{s'}{0}{\lambda t'}}{k_3}{u}$ and an inductive hypothesis for each of those. We also have $\inst{H}{0} {s_0} a {st}$. Now $s_0=s_1 t_1$ must be an application. Note that even in the second rule of \Cref{unfolding}, the definition of $\gg$ on programs implies that the unfolded term would be an abstraction.

So we have $\inst{H}{0}{s_1}a s$ and $\inst{H}{0} {t_1} at$ for some $s_1,t_1$. We now construct the reduction of the machine using the inductive hypothesis. We will explain where the new objects in the following reduction come from in the next paragraph. 
\begin{align}
  (\gamma (s_1t_1)\con P,a)::T,V,H)
  &= (\gamma s_1 \con \gamma t_1 \con \app ::P,a)::T,V,H) &\\
  &\red^{4k_1+1} ((\gamma t_1 \con \app ::P,a)::T,(\gamma s_2,a_2)::V,H_1) & \text{IH} \label{eq:IH-1}\\
  &\red^{4k_2+1} ((\app ::P,a)::T,g_t::(\gamma s_2,a_2)::V,H_2) & \text{IH} \label{eq:IH-2}\\
  &\red  ((\gamma s_2,a_2') ::(P,a)::T,V,H_2') \label{eq:beta}\\
  &\red^{4k_3+1} (([],a_2')::(P,a)::T,g_u::V,H_3) & \text{IH} \label{eq:IH-3}\\
  &\red ((P,a)::T,g_u::V,H_3)
\end{align}

In this reduction, the inductive hypothesis for $s_1$ in (\ref{eq:IH-1}) yields $s_2,a_2$ and $H_1$ such that $H \subseteq H_1$ and $\inst{H_1}{1} {s_2}{a_2}{s'}$. The inductive hypothesis on $t_1$ in (\ref{eq:IH-2}) yields $g_t$ and $H_2$ such that $H_1 \subseteq H_2$ and $g_t \gg_{H_2} \lambda t'$. In the step for beta reduction, (\ref{eq:beta}), we have $(H_2',a_2') = \M{put}{H_2}{\envCons{g_t}{a_2}}$ and $H_2 \subseteq H_2'$. With \Cref{inst-subst}, this implies $\inst{H_2'}{0} {s_2}{a_2'} {\subst{s'}{0}{\lambda t'}}$. This now allows the use of the third inductive hypothesis in (\ref{eq:IH-3}), obtaining $g_u$ and $H_3$ with $g_u \gg u$ and $H_2' \subseteq H_3$. Note that we use \Cref{extend-PO} to transfer several properties along the changing heaps. Now the claim holds for $g_u$ and $H_3$.
\end{mproof}

We define the size of a closure as $\size{(P,a)}:=\size{P}+a$ and the size of a heap entry to be $\size{(g,a)}:=\size{g}+a$. The size of a state is the sum of the sizes of all elements in $T$, $V$ and $H$.

The size of the $k$-th state starting from $\sigma_s$ is a polynomial in $k$ and $\size{s}$:
\begin{technicalappendix}{}
\begin{lemma}[][Analysis]\label{closSpace-detail}
  Assume $\sigma_s \red^k (T,V,H)=\sigma$ for some term $s$.
\begin{enumerate}
  \item $\length{T}+\length{V}\leq k+1$
  \item $\length{H} \leq k$
  \item $\size P \leq \size{s}$ and $a \leq \length H$ for all $P/a \in T \con V$
  \item $\size P \leq \size{s}$ and $a \leq \length H$ and $b \leq \length H$ for all $\envCons{(P,a)}{b} \in H$
  \end{enumerate}
\end{lemma}
\begin{mproof}{closSpace-detail}
All claims follow by induction on $k$. The third claim uses that $\phi P$ always returns a sublist of $P$.
\end{mproof}
\end{technicalappendix}

\begin{theorem}[Heap machine state size][correctSpace]\label{closSpace}
  If $\sigma_s \red^k \sigma$, then $\size \sigma \leq (k+1)(3k+4\size{s})$
\end{theorem}
\begin{mproof}{closSpace}
  Follows from \Cref{closSpace-detail}
\end{mproof}

\section{Simulating L with Turing Machines}\label{sec:L_in_TM}
We now sketch how to construct the Turing machine that simulates \L with polynomially bounded overhead time and constant factor overhead in space.
The considered Turing machines will operate on various kinds of data (e.g. natural numbers, programs, heap closures, heap entries, heaps, \dots).
For programs, we use a symbol for each of the four constructors and a fifth symbol to encode de Bruijn indices in unary.
All other natural numbers will also be encoded in unary, unless explicitly stated.
The encoding of the further structures on tapes is straightforward.

\subsection{The Substitution-based Turing Machine Simulating $\L$}
We construct a Turing machine $M_{\text{subst}}$ that executes the substitution-based strategy from \Cref{sec:subst-machine} for $k$ steps, where $k$ is an input.
The Turing machine takes an additional input~$m$ and aborts if the abstract machine would reach a state of size greater $m$.
\begin{theorem}\label{M_subst}
  There is a Turing machine $M_{\text{\emph{subst}}}$ that, given two binary numbers $k,m$ and a term $s$, halts in time $\bigO{k\cdot\poly{\min{(m,\Space{s})}}}$ and space $\bigO{\min{(m,\Space{s})}+\log{m}+\log{k}}$. Either the machine outputs a term $t$, then $s$ has normal form $t$ and $m\geq \Space{s}$ and $k \geq 3\cdot\Time{s}+1$. Or it halts in one of two other final states: Either a state named \emph{space bound reached}, implying that $m \leq 2\cdot \Space{s}$ holds, or in a state named \emph{space bound not reached}, implying that $k < 3\cdot\Time{s}+1$ holds.
\end{theorem}
\begin{onlyfullversioncomment}{}
Furthermore, by \Cref{subst-space}, the machine can only approximate the size of the 'current' term up to a factor of $2$, which further complicates the theorem.
\end{onlyfullversioncomment}
\begin{proof}%

The Turing machine can be constructed by iterating the rules of the abstract substitution machine from \Cref{fig:subst-red} on the initial state $\tau_s$.
The machine has to keep track of the size of the abstract machine state, even \emph{during} the execution of the substitution:
As soon as the size of the next state to be computed is known to exceed $m$, it aborts \emph{before} consuming more than $\Theta(m)$ space. \begin{onlyfullversioncomment}{}
  This is necessary because the result of a substitution $\subst{P}{0}{Q}$ with $\size P + \size Q \in \bigO{m}$ could have quadratic size $\bigO{m^2}$, e.g. if $P$ applies the variable $0$ to itself $m$ times and $Q$ has size $m$ as well.
The function $\phi P$ can be implemented via the tail-recursive $\phi_{k,Q}P$, which takes space and time $\bigO{\size{Q} + k +\size{P}}$, as it just traverses $P$ and accumulates the result.
The argument $k$ during the run is bound by $\Space{s}$.{ }
\end{onlyfullversioncomment}
Then the size of all intermediate states and the overall space consumption follow from \Cref{subst-space}.
The existence of the result for large enough $k$ follows in combination with \Cref{subst-time}.
\end{proof}
The precise specification of the machine is subtle: 
Intuitively, the machine state size is as large as the 'current' term, but we don't know if a state larger $m$ is reached in the first $k$ steps.
Therefore, we don't specify which of the last two cases occurs if \emph{both} bounds on $k$ and~$m$ are exceeded.

If $s$ diverges, \Cref{M_subst} states that $M_\text{subst}$ can only halt in the two special final states (with $\Time{s}=\infty$ for diverging terms $s$). %

\subsection{The Heap-based Turing Machine Simulating $\L$}
We construct a Turing machine executing the heap-based strategy from Sect.~\ref{sec:heap-machine} for $k$ steps:%
\begin{theorem}\label{M_heap}
  There is a Turing machine $M_{\text{\emph{heap}}}$ that, given a number $k$ and a closed term~$s$, halts in time $\bigO{\poly{\size{s},k}}$ and space $\bigO{\size{s}\cdot\poly{k}}$. If $s$ has a normal form $t$ and $k\geq4\cdot\Time{s}+2$, it computes a heap $H$ and a closure $g$ such that $g\gg_H t$. Otherwise, it halts in a distinguished final state (denoting `\emph{failure}').
\end{theorem}

\begin{proof}
The Turing machine can be constructed by iterating the rule of the abstract substitution machine on the initial state $\sigma_s$.
\begin{onlyfullversioncomment}{}
We already argued on the runtime of $\phi$ for \Cref{M_subst}. And $H[a,n]$ can be computed by iterating over $H$ for at most $n$ times. So each{ }
\end{onlyfullversioncomment}
\begin{onlyicalpversion}
Each{ }
\end{onlyicalpversion}
abstract step $(T,V,H) \red (T',V',H')$ can be implemented in time $\bigO{\poly{\size{(T,V,H)}}}$ and space $\bigO{\max{(\size{(T,V,H)},\size{(T',V',H')})}}$. The space consumption of all involved operations in \Cref{fig:heap-red} is bounded by their input or output. Using \Cref{closSpace}, the size of all intermediate $(T,V,H)$ can be bound by $k$ and $\size{s}$ to derive the claimed resource bounds. The successful computation of $g$ and~$H$ for large enough $k$ follows with \Cref{closTime}.
\end{proof}

\subsection{The Combined Turing Machine Simulating L}
We now combine the machines from the last two sections to execute the heap-machine only if we know that its space consumption is bounded by the space measure of the simulated term:

\begin{theorem}\label{thm:hybrid_simulation}
  There is a Turing machine $M_\text{\L}$ that, given a closed term $s$ that has a normal form $t$, computes a heap $H$ and a closure $g$ such that $g\gg_H t$ in time $\bigO{\poly{\size{s},\Time{s}}}$ and space $\bigO{\Space{s}}$.
\end{theorem}
\begin{proof}\label{M_comb}
  Let $p$ be the polynomial such that the machine from \Cref{M_heap} runs in space $\bigO{\size{s}\cdot p(k)}$.
  Then the combined machine executes the following algorithm:
  \begin{enumerate}
  \item Initialise $k:=0$ (in binary)
  \item\label{M_comb_loop} Compute $m:=\size{s}\cdot p(k)$ (in binary)
  \item\label{M_comb_subst} Run $M_{\text{subst}}$ on $s$, $k$ and $m$.
    \begin{itemize}
    \item If $M_{\text{subst}}$ computes the normal form $t$, output $(\gamma{t},0)$ and an empty heap $[]$ and halt.
    \item If $M_{\text{subst}}$ halts with \emph{space bound not reached}, set $k:=k+1$ and go to \textbf{\ref{M_comb_loop}}.
    \item If $M_{\text{subst}}$ halts with \emph{space bound reached}, continue at \textbf{\ref{M_comb_heap}}.
    \end{itemize}
  \item\label{M_comb_heap} Run $M_{\text{heap}}$ on $s$ and $k$.
    \begin{itemize}
    \item If this computed a closure and a heap representing $t$, output that and halt.
    \item Otherwise, set $k:=k+1$ and go to \textbf{\ref{M_comb_loop}}.
    \end{itemize}
  \end{enumerate}
  First, we show that if this machine halts, its output is a closure-heap pair representing the normal form $t$ of $s$: If the machine halts during \textbf{\ref{M_comb_subst}}, the output is a representation of the normal form by \Cref{M_subst} and \Cref{bound-inst}. If it halts during \textbf{\ref{M_comb_heap}}, it does so by \Cref{M_heap}.
  
  Second, we analyse termination and the time complexity of this machine. As intermediate step,  we analyse the run time for a fixed $k$.
  Step \textbf{\ref{M_comb_loop}} takes time $\bigO{\poly{\size s,k}}$, and the size of~$s$ can be computed from its encoding in straightforward fashion. Using \Cref{M_subst}, Step \textbf{\ref{M_comb_subst}} takes time
\begin{onlyfullversioncomment}{}
  \begin{align*}
    \bigO{k\cdot\poly{\min{(m,\Space{s})}}}
    &\subseteq\bigO{k\cdot\poly{m}}\\
    &=\bigO{k\cdot\poly{\size{s}\cdot p(k)}}\\
    &\subseteq\bigO{k\cdot\poly{\size{s},k}}&&\text{$p$ \emph{is} a polynomial}\\
    &\subseteq \bigO{\poly{\size{s},k}}
  \end{align*}
\end{onlyfullversioncomment}
\begin{onlyicalpversion}
  \begin{align*}
    \bigO{k\cdot\poly{\min{(m,\Space{s})}}}
    \subseteq\bigO{k\cdot\poly{m}}
    =\bigO{k\cdot\poly{\size{s}\cdot p(k)}}
    \subseteq \bigO{\poly{\size{s},k}}
  \end{align*}
\end{onlyicalpversion}
  If Step \textbf{\ref{M_comb_heap}} is executed, this takes time $\bigO{\poly{\size{s},k}}$ by \Cref{M_heap}. This means for arbitrary~$k$, one iteration of the described algorithm can be computed in time $\bigO{\poly{\size s,k}}$.
  
  The algorithm will eventually halt: We consider $k=4\Time{s}+2$, which is larger than the two values required in \Cref{M_subst} and \Cref{M_heap}: By \Cref{M_subst}, the machine does halt during Step \textbf{\ref{M_comb_subst}}, unless $m < \Space{s}$. In the latter case, \textbf{\ref{M_comb_heap}} is tried. Then, by \Cref{M_heap}, as $k$ is large enough, we have that $M_\text{heap}$ indeed halts with a closure-heap pair.

  Summing up the run time of each iteration, we have that the machine terminates in time
  \begin{align*}
    \bigO{\sum_{k=0}^{4\Time{s}+2}\left( \poly{\size s,k} \right)}\subseteq \bigO{\Time{s}\cdot( \poly{\size s,\Time{s}})} \subseteq \bigO{\poly{\size s,\Time{s}}}
  \end{align*}

  Third, we analyse the space complexity
  of this machine. Again, we first analyse one iteration for a fixed $k$.
  Step \textbf{\ref{M_comb_loop}} takes space $\bigO{\log(m)}$, since we use binary numbers.
  By \Cref{M_subst}, Step \textbf{\ref{M_comb_subst}} takes space $\bigO{\min{(m,\Space{s})}+\log{m}+\log{k}} \subseteq \bigO{\Space{s}+\log{m}+\log{k}}$.
  If Step \textbf{\ref{M_comb_heap}} is executed, then $m < \Space{s}$.
  By \Cref{M_heap}, this step runs in space $\bigO{m} \subseteq \bigO{\Space{s}}$.
  \begin{onlyfullversioncomment}{}
    So, we can compute the space consumption of a single iteration as:
  \begin{align*}
    &\bigO{\log m + \log k + \Space{s}} \\
    &= \bigO{\log{(\size{s}\cdot p(k))}+\log{k}+\Space{s}} && \text{definition $m$}\\
    & \subseteq \bigO{\log{\size{s}} + \log (p(k)) +\log{k} + \Space{s}} && \begin{onlyinternaldetails}{} \text{logarithm of products} \end{onlyinternaldetails}\\
    & = \bigO{\log (p(k)) +\log{k} + \Space{s}} && \text{as $\size{s} \leq \Space{s}$}\\
    & \subseteq \bigO{\log{k} + \Space{s}} && \text{$\log (p(k)) \in \bigO{\log{k}}$ as $p$ polynomial}                     
  \end{align*}
\end{onlyfullversioncomment}
\begin{onlyicalpversion}{}
    Since $\bigO{\log m} = \bigO{\log (\size{s}\cdot p(k))} \subseteq \bigO{\log{\size{s}} + \log (p(k))} \subseteq \bigO{ \Space{s} + \log{k}}$, the space consumption of a single iteration is $\bigO{\log m + \log k + \Space{s}} \subseteq \bigO{\log{k} + \Space{s}}$.
\end{onlyicalpversion}
  
 Overall, we have that the whole machine runs in space (the last equation is by \Cref{space-bounds-time}):
 \[
    \bigO{\max_{0 \leq k \leq 4\Time{s}+2}\left( \log{k} + \Space{s} \right)}\subseteq \bigO{ \log{\Time{s}}+\Space{s}}
    = \bigO{\Space{s}} \hfill\popQED
  \]
\end{proof}

Note that the machine only terminates for terminating terms, making this a full simulation also for diverging terms.
For terms with $\Time{s} \not\in \bigO{ \Space{s}}$ it is crucial that the machine tracks the step number $k$ in binary, because it would need $\Omega{\Time s}$ space otherwise.
This suffices due to the following theorem, which is proved in the appendix:

\begin{theorem}\label{space-bounds-time}
  $\log{\Time{s}} \in \bigO{\Space{s}}$.
\end{theorem}

\begin{mproof}{space-bounds-time}
  The main insight is that for any given size, there are only exponentially many terms smaller than that size.
  As reduction is deterministic, a terminating term $s$ can not contain the same intermediate term twice.
  This bounds $\Time{s}$ by the number of terms with size smaller than $\Space{s}$, i.e.   $\Time{s} \leq c^{\Space{s}}$ for a constant $c$.%

  Now, we show that the number of terms smaller than a certain size $m$ is an exponential. We use the encoding $\gamma$ to allow us to count linear strings (programs) instead of trees (terms):
  \begin{align*}
    &\#\{t\mid \size{t} \leq m\}\\
    &= \#\{\ t       \mid 2\cdot\size{t} \leq 2 \cdot m\}\\
    &\leq \#\{ t     \mid \size{\gamma{t}} \leq 2 \cdot m\} && \text{\Cref{size-gamma}}\\
    &= \#\{\gamma t  \mid \size{\gamma{t}} \leq 2 \cdot m\} && \text{\Cref{gamma-inj}}\\
    &\leq \#\{P  \mid \size{P} \leq 2 \cdot m\}\\
    &\leq 5^{2\cdot m}
  \end{align*}
  In the last step, we use that $\#\{P  \mid \size{P} \leq n\} \leq 5^{n-1}$ for all $n>0$ by induction on $n$, where the intuition behind the $5$ is that there are four different symbols with which a program can start, and that variables require a fifth symbol to encode the index in unary.
Thus the claim holds for  $c=5^2$.  
\end{mproof}

The simulation of L on Turing machines computes normal form as pair of closure and heap, as defined in \Cref{rep-rel-clos}.
It is possible to unfold this heap into a program:

\begin{lemma}\label{lem:declosure}
  There is a machine $M_{\text{\emph{unf}}}$ that, given a heap $H$ and a closure $g$ that represent~$s$, i.e.\ $g \gg_H s$, computes $s$ (explicitly encoded as $\gamma s$) in time $\bigO{\poly{\size{s},\size{H},\size{g}}}$ and space $\bigO{\size{s}\cdot(\size{g}+\size{H})}$.
\end{lemma}
\begin{mproof}{lem:declosure}
  We first consider a partial function $f_H P a k$ that computes the unfolding, but on programs instead of terms:

  \begin{align*}
    f_H [] a k ~&:=~ [] \\
    f_H (\app::P) a k ~&:=~ \app::f_H P a k \\
   f_H (\ret::P) a (\natS k)~&:=~ \ret::f_H P a k \\
    f_H (\lamb::P) a k~&:=~ \lamb::f_H P a (\natS k) \\
    f_H (\var\,n::P) a k~&:=~ \lamb::f_H Q b 1 \con \ret::f_H P a k && \text{if $n \geq k$ and $H[a,n-k]=(Q,b)$} \\
    f_H (\var\,n::P) a k ~&:=~ \var\,n ::f_H P a k && \text{if $n < k$}
  \end{align*}
  For this set of equations, we can show \emph{if $\inst H a k s {s'}$, then $f_H (\gamma s \con Q) a k = \gamma {s'} \con f_H Q a k $} by induction on $s'$. With $Q=[]$, this means that $(P,a) \gg_H \lambda s'$ implies $f_H P a 1 = \gamma {s'}$. So $f$ indeed computes the unfolding on programs.

  Implementing $f$ in Turing machines, we first note that during execution, all considered $k$, $a$ and $P$ are bound is bound by $\size{g}+\size{H}$, as all addresses come from $g$ or $H$ and $k$ can not be larger than the largest program in $g$ or $H$. In the equation using $H[\cdot,\cdot]$, an additional explicit stack is needed to remember $P$ for after the recursive call on $Q$. This stack is bound in length by $\bigO{\size{s}}$, as every recursive call computes at least one symbol of the result. 
 This means that the algorithm runs in space $\bigO{\size{s}\cdot(\size{g}+\size{H})}$.

Furthermore, each equation produces a symbol of the result, so the total number of calls on $f$ is bound by $\size{s}$. Every equation, except the one where $H[\cdot,\cdot]$ occurs, performs a constant number of operations. This other equation needs to traverse $H$ at most $n$ times before recurring, where $n$ is the largest de Bruijn index occurring. 
  In total, this means that the algorithm runs in time $\bigO{\poly{\size{s},\size{g},\size{H}}}$.
\end{mproof}

\section{Simulating Turing Machines in L}\label{sec:TM_in_L}
\begin{onlyicalpversion}
  \FK{TODO:move to appendix}
\end{onlyicalpversion}
The remaining direction of the proof of the strong invariance thesis requires us to prove that Turing machines can be simulated with $\L$ consuming only a constant overhead in space and a polynomial overhead in time with respect to our measures $\Space{\cdot}$ and $\Time{\cdot}$.

Accattoli and Dal Lago~\cite{ADL} show that counting head-reductions is an invariant time measure.
In the associated technical report, they give a linear simulation of Turing machines in the deterministic $\lambda$-calculus, a fragment of the $\lambda$-calculus where all weak evaluation strategies coincide.
Although they treat variables as values, reduction in L also coincides, because all considered terms are closed.
The construction uses standard Scott encodings $\encode{x}$ for strings $x$ and is explained in all detail in~\cite{lago_encoding_2017}, spelling out all intermediate terms during simulation explicitly.

It turns out that this construction also only has a constant factor overhead in space w.r.t our measure $\Space{\cdot}$.
This can easily be verified by checking all intermediate terms spelled out in the proofs of~\cite{lago_encoding_2017}.
One has to take care that a linear amount of steps (i.e. all steps annotated with $\bigO{\cdot}$ or $\Theta(\cdot)$ instead of constants) does not introduce a super-linear space overhead.
This is the case, because all such sequences of steps only use substitutions where the substituted variable occurs at most once, effectively decreasing the term size.
Note that since names in the simulation are all distinct, the translation to de Bruijn indices has no overhead.
Thus the simulation is linear in time \textit{and space}:

\begin{theorem}\label{thm:easy_direction}
Let $f:\Sigma^\ast \rightarrow \Sigma^\ast$ be a function that is computable by a Turing machine $\mathcal{M}$ in time $\mathcal{T}$ and in space $\mathcal{S}$. Then there exists an $\L$-term $\overline{\mathcal{M}}$ such that for every $x \in \Sigma^\ast$ we have that
\begin{enumerate}
\item $\overline{\mathcal{M}}~\encode{x} \red^\ast \encode{f(x)}$,
\item $\Space{\overline{\mathcal{M}}~\encode{x}} \in \bigO{|x|+\mathcal{S}(|x|)}$, and
\item $\Time{\overline{\mathcal{M}}~\encode{x}} \in \bigO{|x|+\mathcal{T}(|x|)}$.
\end{enumerate}
\end{theorem}
\begin{proof}
  Take $\overline{\mathcal{M}}$ as in Theorem 5.5. in~\cite{lago_encoding_2017}.
\end{proof}

\section{The Weak Call-By-Balue $\lambda$-Calculus is Reasonable}\label{sec:reasonable}

We explain how existing simulations of Turing machine in the $\lambda$-calculus already have polynomial time and constant factor space overhead in~\ref{sec:TM_in_L}.
With both the simulations, we are now able to show Theorem~\ref{thm:intro_main}, that is, the invariance thesis for the weak call-by-value $\lambda$-calculus.
\begin{onlyicalpversion} We give the full proof in~\Cref{sec:fullproof} and only sketch it here:
\begin{proof}[Proof Sketch of Theorem 2]
  Let $s_f$ compute $f$.
  We construct a Turing machine $M_f$ which first executes $M_{\textsf{L}}$ (Theorem~\ref{thm:hybrid_simulation}) on $s_f ~ \encode {x}$ and then uses $M_{\text{unf}}$ (Lemma~\ref{lem:declosure}) to distinguish the normal form $\encode \True$ from $\encode \False$.
  Since the size of both $\encode \True$ and $\encode \False$ is a constant $b$ only depending on the size of the (fixed) alphabet, we obtain the time consumption $\bigO{\poly{\size{s},\Time{s}}+ \poly{b, \size{H},\size{g}}} \subseteq \bigO{\poly{\mathcal{T}(|x|)}}$ and space consumption
  $\bigO{\Space{s} + b \cdot (\size{H} + \size{g})} \subseteq \bigO{\mathcal{S}(|x|)}$.
  
  For the converse direction, let $\mathcal M$ compute $f$ in time $\mathcal{T}$ and space $\mathcal{S}$.
  By~\Cref{thm:easy_direction}, $\overline{\mathcal{M}}$ computes $f$ in space $\bigO{|x|+\mathcal{S}(|x|)}$, and time $\bigO{|x|+\mathcal{T}(|x|)}$.
  We have $\bigO{|x|+\mathcal{S}(|x|)} = \bigO{\mathcal{S}(|x|)}$ and $\bigO{\mathcal{T}(|x|)} =\bigO{|x|+\mathcal{T}(|x|)}$ as both $\mathcal{S}$ and~$\mathcal{T}$ are contained in $\Omega(n)$.
\end{proof}
\end{onlyicalpversion}
\begin{proof}
Let $\Sigma$ be a finite alphabet such that $\{\True,\False\}\subseteq \Sigma$ and let $f:\Sigma^\ast \rightarrow \{\True,\False\}$ be a function. Furthermore, let $b = \max\{\size{\encode{\True}},\size{\encode{\False}}\}$ and $\mathcal{T},\mathcal{S}\in \Omega(n)$. Note that $b$ is a constant only depending on the fixed alphabet $\Sigma$.

For the first direction, we assume that $f$ is $\L$-computable in time $\mathcal{T}$ and space $\mathcal{S}$. By definition, there is hence a term $s_f$ such that for all $x \in \Sigma^\ast$ we have that
\[s_f \encode{x} \red^\ast \encode{f(x)}~~\text{and}~~ \Time{s_f \encode x} \leq \mathcal{T}(|x|) ~~\text{and} ~~ \Space{s_f \encode x} \leq \mathcal{S}(|x|) \,.\]
We construct a Turing machine $M_f$ as follows. On input $x$, $M_f$ executes $M_\L$ on the (closed) term $s := s_f \encode{x}$, which computes a heap $H$ and a closure $g$ such that  $g\gg_H \encode{f(x)}$ in time $\bigO{\poly{\size{s},\Time{s}}}$ and space $\bigO{\Space{s}}$, by Theorem~\ref{thm:hybrid_simulation} -- note that $s_f$ as well as $M_\L$ are hard-coded in $M_f$. We observe that
\begin{equation}\label{eq:surprisingly_important}
\size{g}+\size{H} \in \bigO{\poly{\size{s},\Time{s}}}~~~\text{ and }~~~ \size{g}+\size{H} \in \bigO{\Space{s}} \,,
\end{equation}
where the former holds as writing down $g$ and $H$ cannot take more time than the overall running time bound $\bigO{\poly{\size{s},\Time{s}}}$ and the latter is due to the space bound $\bigO{\Space{s}}$ of~$M_f$.
After that, $M_f$ executes $M_\text{unf}$ on $H$ and $g$ which yields $\encode{f(x)}$ and finally, depending on whether $\encode{f(x)}=\encode{\True}$ or $\encode{f(x)}=\encode{\False}$, $M_f$ outputs $\True$ or $\False$ accordingly. By Lemma~\ref{lem:declosure}, the final steps take time $\bigO{\poly{b, \size{H},\size{g}}}$ and space $\bigO{b\cdot(\size{g}+\size{H})}$. Now the final time consumption is given by

\begin{align}
\bigO{\poly{\size{s},\Time{s}}+ \poly{b, \size{H},\size{g}}} 
& \leq \bigO{\poly{\size{s},\Time{s}}+ \poly{b,\size{s},\Time{s}}} \label{eq:m1}\\
~&\leq \bigO{\poly{|x|,\mathcal{T}(|x|)}} \label{eq:m2} \\
~&\leq \bigO{\poly{\mathcal{T}(|x|)}}\,, \label{eq:m3}
\end{align}
where (\ref{eq:m1}) is due to Equation~(\ref{eq:surprisingly_important}), (\ref{eq:m2}) holds as $\size{s_f}$ and $b$ are constants and (\ref{eq:m3}) follows from the fact that $\mathcal{T}\in \Omega(n)$.
\noindent The overall space consumption is bounded by
\begin{align}
\bigO{\Space{s} + b \cdot (\size{H} + \size{g})} 
& \leq \bigO{(b+1) \cdot \Space{s}} \label{eq:m4}\\
~&\leq \bigO{\mathcal{S}(|x|)} \label{eq:m5} \,,
\end{align}
where (\ref{eq:m4}) is due to Equation~(\ref{eq:surprisingly_important}) and (\ref{eq:m5}) holds as $b$ is a constant.\\
$~$

For the converse direction, we assume that $f$ can be computed by a Turing machine $\mathcal{M}$ in time $\mathcal{T}$ and space $\mathcal{S}$. We invoke Theorem~\ref{thm:easy_direction} to obtain a term $\overline{\mathcal{M}}$ which shows that $f$ is $\L$-computable in space $\bigO{|x|+\mathcal{S}(|x|)}$, and time $\bigO{|x|+\mathcal{T}(|x|)}$. We conclude the proof by observing that $\bigO{|x|+\mathcal{S}(|x|)} = \bigO{\mathcal{S}(|x|)}$ and $\bigO{\mathcal{T}(|x|)} =\bigO{|x|+\mathcal{T}(|x|)}$ as both, $\mathcal{S}$ and~$\mathcal{T}$ are contained in $\Omega(n)$.
\end{proof}

\section{Related and Future Work}

We have already mentioned the recent long line of work by Accattoli, Dal Lago, Sacerdoti Coen,  Guerriri and Martini (for an overview see~\cite{accattoli2018efficiency}) analysing reasonable time measures and implementations of several $\lambda$-calculi.

Type systems for call-by-name and call-by-value $\lambda$-calculi can be used to logically characterise complexity classes  (P~\cite{asperti2002intuitionistic}, LOGSPACE~\cite{10.1007/11874683_40}, PSPACE~\cite{gaboardi2008logical}).
Connecting these insights with our measures would make it even more feasible to use \L as a formal basis for complexity theory, which we plan to do as future work, building on existing formalisations of computability theory~\cite{Forster17}.

There is recent work in investigating strategies to evaluate open terms, for instance open call-by-value, which is reasonable for time~\cite{accattoli2017implementing, accattoli2015relative}, but the question for space is open.
On the more applied side, there is work on time and space profiling based on lazy graph reduction~\cite{sansom1995time} in Haskell.
More recent work uses a graph-based cost-semantics used for space-profiling~\cite{spoonhower_space_2008}, based on earlier measures in~\cite{blelloch1995parallelism}.
Moreover, computation in sub-linear space with an external memory has been studied~\cite{dal_lago_functional_2010}, which we do not cover in this paper.

And finally, the full $\lambda$-calculus can be translated into weak call-by-value e.g. using a CPS translation.
The longstanding question whether the natural time and space measures for the $\lambda$-calculus are reasonable remains open.
We want to investigate whether our results can contribute to an answer.

\bibliography{bib}

\begin{thebibliography}{10}

\bibitem{accattoli2017complexity}
Beniamino Accattoli.
\newblock The {C}omplexity of {A}bstract {M}achines.
\newblock In {\em Proceedings Third International Workshop on Rewriting
  Techniques for Program Transformations and Evaluation, WPTE@FSCD 2016, Porto,
  Portugal, 23rd June 2016.}, pages 1--15, 2016.
\newblock \href {http://dx.doi.org/10.4204/EPTCS.235.1}
  {\path{doi:10.4204/EPTCS.235.1}}.

\bibitem{accattoli2018efficiency}
Beniamino Accattoli.
\newblock ({I}n){E}fficiency and {R}easonable {C}ost {M}odels.
\newblock {\em Electr. Notes Theor. Comput. Sci.}, 338:23--43, 2018.
\newblock \href {http://dx.doi.org/10.1016/j.entcs.2018.10.003}
  {\path{doi:10.1016/j.entcs.2018.10.003}}.

\bibitem{accattoli2015relative}
Beniamino Accattoli and Claudio~Sacerdoti Coen.
\newblock {O}n the {R}elative {U}sefulness of {F}ireballs.
\newblock In {\em 30th Annual {ACM/IEEE} Symposium on Logic in Computer
  Science, {LICS} 2015, Kyoto, Japan, July 6-10, 2015}, pages 141--155, 2015.
\newblock \href {http://dx.doi.org/10.1109/LICS.2015.23}
  {\path{doi:10.1109/LICS.2015.23}}.

\bibitem{accattoli2017implementing}
Beniamino Accattoli and Giulio Guerrieri.
\newblock Implementing {O}pen {C}all-by-{V}alue.
\newblock In {\em Fundamentals of Software Engineering - 7th International
  Conference, {FSEN} 2017, Tehran, Iran, April 26-28, 2017, Revised Selected
  Papers}, pages 1--19, 2017.
\newblock \href {http://dx.doi.org/10.1007/978-3-319-68972-2\_1}
  {\path{doi:10.1007/978-3-319-68972-2\_1}}.

\bibitem{ADL}
Beniamino Accattoli and Ugo~Dal Lago.
\newblock On the {I}nvariance of the {U}nitary {C}ost {M}odel for {H}ead
  {R}eduction.
\newblock In {\em 23rd International Conference on Rewriting Techniques and
  Applications (RTA'12) , {RTA} 2012, May 28 - June 2, 2012, Nagoya, Japan},
  pages 22--37, 2012.
\newblock \href {http://dx.doi.org/10.4230/LIPIcs.RTA.2012.22}
  {\path{doi:10.4230/LIPIcs.RTA.2012.22}}.

\bibitem{indeed}
Beniamino Accattoli and Ugo~Dal Lago.
\newblock ({L}eftmost-{O}utermost) {B}eta {R}eduction is {I}nvariant, {I}ndeed.
\newblock {\em Logical Methods in Computer Science}, 12(1), 2016.
\newblock \href {http://dx.doi.org/10.2168/LMCS-12(1:4)2016}
  {\path{doi:10.2168/LMCS-12(1:4)2016}}.

\bibitem{asperti2002intuitionistic}
Andrea Asperti and Luca Roversi.
\newblock Intuitionistic {L}ight {A}ffine {L}ogic.
\newblock {\em {ACM} Trans. Comput. Log.}, 3(1):137--175, 2002.
\newblock \href {http://dx.doi.org/10.1145/504077.504081}
  {\path{doi:10.1145/504077.504081}}.

\bibitem{blelloch1995parallelism}
Guy~E. Blelloch and John Greiner.
\newblock Parallelism in {S}equential {F}unctional {L}anguages.
\newblock In {\em Proceedings of the seventh international conference on
  Functional programming languages and computer architecture, {FPCA} 1995, La
  Jolla, California, USA, June 25-28, 1995}, pages 226--237, 1995.
\newblock \href {http://dx.doi.org/10.1145/224164.224210}
  {\path{doi:10.1145/224164.224210}}.

\bibitem{deBruijn1972}
Nicolaas Govert~De Bruijn.
\newblock Lambda calculus notation with nameless dummies, a tool for automatic
  formula manipulation, with application to the {C}hurch-{R}osser theorem.
\newblock In {\em Indagationes Mathematicae (Proceedings)}, volume~75, pages
  381--392. Elsevier, 1972.

\bibitem{dershowitz_invariance_nodate}
Nachum Dershowitz and Evgenia Falkovich-Derzhavetz.
\newblock The invariance thesis.
\newblock {\em Logical Methods in Computer Science (to appear)}, 2015.
\newblock URL:
  \url{http://www.cs.tau.ac.il/~nachumd/papers/InvarianceThesis.pdf}.

\bibitem{LOLA}
Yannick Forster, Fabian Kunze, and Marc Roth.
\newblock The strong invariance thesis for a $\lambda$-calculus.
\newblock {\em Workshop on Syntax and Semantics of Low-Level Languages (LOLA)},
  2017.

\bibitem{Forster17}
Yannick Forster and Gert Smolka.
\newblock Weak {C}all-by-{V}alue {L}ambda {C}alculus as a {M}odel of
  {C}omputation in {C}oq.
\newblock In {\em Interactive Theorem Proving - 8th International Conference,
  {ITP} 2017, Bras{\'{\i}}lia, Brazil, September 26-29, 2017, Proceedings},
  pages 189--206, 2017.
\newblock \href {http://dx.doi.org/10.1007/978-3-319-66107-0\_13}
  {\path{doi:10.1007/978-3-319-66107-0\_13}}.

\bibitem{gaboardi2008logical}
Marco Gaboardi, Jean{-}Yves Marion, and Simona Ronchi~Della Rocca.
\newblock A logical account of {P}{S}{P}{A}{C}{E}.
\newblock In {\em Proceedings of the 35th {ACM} {SIGPLAN-SIGACT} Symposium on
  Principles of Programming Languages, {POPL} 2008, San Francisco, California,
  USA, January 7-12, 2008}, pages 121--131, 2008.
\newblock \href {http://dx.doi.org/10.1145/1328438.1328456}
  {\path{doi:10.1145/1328438.1328456}}.

\bibitem{cbvlcm2}
Fabian Kunze, Gert Smolka, and Yannick Forster.
\newblock Formal {S}mall-{S}tep {V}erification of a {C}all-by-{V}alue {L}ambda
  {C}alculus {M}achine.
\newblock In {\em Programming Languages and Systems - 16th Asian Symposium,
  {APLAS} 2018, Wellington, New Zealand, December 2-6, 2018, Proceedings},
  pages 264--283, 2018.
\newblock \href {http://dx.doi.org/10.1007/978-3-030-02768-1\_15}
  {\path{doi:10.1007/978-3-030-02768-1\_15}}.

\bibitem{lago_encoding_2017}
Ugo~Dal Lago and Beniamino Accattoli.
\newblock Encoding {T}uring {M}achines into the {D}eterministic
  {L}ambda-{C}alculus.
\newblock {\em CoRR}, abs/1711.10078, 2017.
\newblock URL: \url{http://arxiv.org/abs/1711.10078}, \href
  {http://arxiv.org/abs/1711.10078} {\path{arXiv:1711.10078}}.

\bibitem{DalLagoMartini08}
Ugo~Dal Lago and Simone Martini.
\newblock The weak lambda calculus as a reasonable machine.
\newblock {\em Theor. Comput. Sci.}, 398(1-3):32--50, 2008.
\newblock \href {http://dx.doi.org/10.1016/j.tcs.2008.01.044}
  {\path{doi:10.1016/j.tcs.2008.01.044}}.

\bibitem{dal_lago_functional_2010}
Ugo~Dal Lago and Ulrich Sch{\"{o}}pp.
\newblock Functional {P}rogramming in {S}ublinear {S}pace.
\newblock In {\em Programming Languages and Systems, 19th European Symposium on
  Programming, {ESOP} 2010, Held as Part of the Joint European Conferences on
  Theory and Practice of Software, {ETAPS} 2010, Paphos, Cyprus, March 20-28,
  2010. Proceedings}, pages 205--225, 2010.
\newblock \href {http://dx.doi.org/10.1007/978-3-642-11957-6\_12}
  {\path{doi:10.1007/978-3-642-11957-6\_12}}.

\bibitem{lawall1996optimality}
Julia~L. Lawall and Harry~G. Mairson.
\newblock Optimality and {I}nefficiency: What {I}sn't a {C}ost {M}odel of the
  {L}ambda {C}alculus?
\newblock In {\em Proceedings of the 1996 {ACM} {SIGPLAN} International
  Conference on Functional Programming, {ICFP} 1996, Philadelphia,
  Pennsylvania, USA, May 24-26, 1996.}, pages 92--101, 1996.
\newblock \href {http://dx.doi.org/10.1145/232627.232639}
  {\path{doi:10.1145/232627.232639}}.

\bibitem{Norrish2011}
Michael Norrish.
\newblock Mechanised {C}omputability {T}heory.
\newblock In {\em Interactive Theorem Proving - Second International
  Conference, {ITP} 2011, Berg en Dal, The Netherlands, August 22-25, 2011.
  Proceedings}, pages 297--311, 2011.
\newblock \href {http://dx.doi.org/10.1007/978-3-642-22863-6\_22}
  {\path{doi:10.1007/978-3-642-22863-6\_22}}.

\bibitem{Plotkin75}
Gordon~D. Plotkin.
\newblock Call-by-{N}ame, {C}all-by-{V}alue and the lambda-{C}alculus.
\newblock {\em Theor. Comput. Sci.}, 1(2):125--159, 1975.
\newblock \href {http://dx.doi.org/10.1016/0304-3975(75)90017-1}
  {\path{doi:10.1016/0304-3975(75)90017-1}}.

\bibitem{sansom1995time}
Patrick~M. Sansom and Simon~L. {Peyton Jones}.
\newblock Time and {S}pace {P}rofiling for {N}on-{S}trict {H}igher-{O}rder
  {F}unctional {L}anguages.
\newblock In {\em Conference Record of POPL'95: 22nd {ACM} {SIGPLAN-SIGACT}
  Symposium on Principles of Programming Languages, San Francisco, California,
  USA, January 23-25, 1995}, pages 355--366, 1995.
\newblock \href {http://dx.doi.org/10.1145/199448.199531}
  {\path{doi:10.1145/199448.199531}}.

\bibitem{10.1007/11874683_40}
Ulrich Sch{\"{o}}pp.
\newblock Space-{E}fficient {C}omputation by {I}nteraction.
\newblock In {\em Computer Science Logic, 20th International Workshop, {CSL}
  2006, 15th Annual Conference of the EACSL, Szeged, Hungary, September 25-29,
  2006, Proceedings}, pages 606--621, 2006.
\newblock \href {http://dx.doi.org/10.1007/11874683\_40}
  {\path{doi:10.1007/11874683\_40}}.

\bibitem{Slot:1984:TVC:800057.808705}
Cees~F. Slot and Peter van Emde~Boas.
\newblock On {T}ape {V}ersus {C}ore; {A}n {A}pplication of {S}pace {E}fficient
  {P}erfect {H}ash {F}unctions to the {I}nvariance of {S}pace.
\newblock In {\em Proceedings of the 16th Annual {ACM} Symposium on Theory of
  Computing, April 30 - May 2, 1984, Washington, DC, {USA}}, pages 391--400,
  1984.
\newblock \href {http://dx.doi.org/10.1145/800057.808705}
  {\path{doi:10.1145/800057.808705}}.

\bibitem{spoonhower_space_2008}
Daniel Spoonhower, Guy~E. Blelloch, Robert Harper, and Phillip~B. Gibbons.
\newblock Space profiling for parallel functional programs.
\newblock {\em J. Funct. Program.}, 20(5-6):417--461, 2008.
\newblock \href {http://dx.doi.org/10.1017/S0956796810000146}
  {\path{doi:10.1017/S0956796810000146}}.

\bibitem{Coq}
{The Coq Proof Assistant}.
\newblock \url{http://coq.inria.fr}, 2018.

\end{thebibliography}

\newpage

\appendix
\def\inappendix{}

\begin{onlyicalpversion}
\section{Proof of \Cref{thm:intro_main}}\label{sec:fullproof}
\FK{TODO: move proof back}
\end{onlyicalpversion}

\section{Big-step characterisation of reduction}
\setCoqFilename{LM.L}
While the characterisation of our time and space measure in terms of a normalising reduction $s_0 \red \ldots \red s_k$ are intuitive, a big-step characterisation allows for easy, inductive analyses of the abstract machines evaluating $\L$ in \Cref{sec:abstract_machines}. 
  
\begin{definition}[Time Measure][timeBS]
  \label{def:time-bs}
  \begin{mathpar}
  \inferrule*{~}{\TimeBS{\lambda s}{0}{\lambda s}}
  \and
  \inferrule*{\TimeBS{s}{k_1}{\lambda s'}
    \and \TimeBS{t}{k_2}{\lambda t'}
    \and \TimeBS{\subst{s'}{0}{\lambda t'}}{k_3}{u}}
  {\TimeBS{s t}{k_1+k_2+1+k_3}{u}}
\end{mathpar}
\end{definition}

\begin{definition}[Space Measures][spaceBS]
  \label{def:space-bs}
  \begin{mathpar}
  \inferrule*{~}{\SpaceBS{\lambda s}{|\lambda s|}{\lambda s}}
  \and
  \inferrule*{\SpaceBS{s}{m_1}{\lambda s'}
    \and \SpaceBS{t}{m_2}{\lambda t'}
    \and \SpaceBS{\subst{s'}{0}{\lambda t'}}{m_3}{u}
    \and m = \max (1+m_1+|t|,1+|\lambda s'|+m_2,m_3)}
  {\SpaceBS{s t}{m}{u}}
\end{mathpar}
\end{definition}
In the second rule, each of the three recursive assumptions could contain the largest subterm, so we take the maximum of each $m_i$, while accounting for the size of the remaining part of the term, e.g.\ during the reduction of $s$ in $s t$, the $t$ and the application itself contribute to $1 + \size{t}$ additional size by definition of the term siz.

The following lemmas allow us to use the two characterisations interchangeably:

\begin{lemma}[][timeBS_correct]\label{lem-bigstep-time}
  $\TimeBS{s}{k}{t}$ iff $s \red^k t$ and $t$ is an abstraction.
\end{lemma}

\begin{mproof}{lem-bigstep-time}
  For the direction assuming $\TimeBS{s}{k}{t}$, the claim follows by induction on $\TimeBS{}{}{}$using two compatibility lemmas of $\red^k$ with term-level application, the first beeing that $s \red^k s'$ implies $s t \red^k s' t$, and the second that $t \red^k t'$ implies $(\lambda s) t \red^k (\lambda s) t'$.

  For the other direction, we first show
  \begin{claim}\label{step-TimeBS}
    If $s \red s'$ and $\TimeBS{s'}{k}{t}$, then $\TimeBS{s}{\natS k}{t}$.
  \end{claim}
  This claim follows by induction in $s \red s'$.

  Now, assuming $s \red^{k} \lambda t$, we can show $\TimeBS{s}{k}{t}$ by induction on $k$ using \Cref{step-TimeBS} in the case where $k>0$.
\end{mproof}

Note that especially if $\Time{s}=k$, then $\TimeBS{s}{k}{t} $ for some $t$.

We write $s \red^*_m t$ if $s$ reduces to $t$ where the largest intermediate term has size $m$.

\begin{lemma}[][spaceBS_correct]\label{lem-bigstep-space}
  $\SpaceBS{s}{m}{t}$ iff $s \red^*_m t$ and $t$ is an abstraction.
\end{lemma}
\begin{mproof}{lem-bigstep-space}
  For the direction assuming $\SpaceBS{s}{m'}{t}$, the claim follows by induction on $\SpaceBS{}{}{}$. In the inductive case, two compatibility lemmas of $\red^*_{m}$ with term-level application are helpful: the first beeing that $s \red^*_{m} s'$ implies $s t \red^*_{1 + m + \size{t}} s' t$, and the second that $t \red^*_{m} t'$ implies $(\lambda s) t \red^*_{1+ m +\size{\lambda s}} (\lambda s) t'$. Furthermore, the fact that $\SpaceBS{t}{m_2}{\lambda t'}$ implies $\size{\lambda t'} \leq m_2$ is needed.

  For the other direction, we first show
  \begin{claim}\label{step-SpaceBS}
    If $s \red s'$ and $\SpaceBS{s'}{m}{t}$, then $\SpaceBS{s}{\max(\size s,m)}{t}$.
  \end{claim}
  This claim follows by induction in $s \red s'$. The equalities between expressions involving $\max$ are tedious to check, but follow only using the inductive hypothesis, the definition of $\SpaceBS{}{}{}$, the definition of the size of terms and that $\SpaceBS{s}{m}{t}$ implies $\size s \leq m \geq \size t$.

  Now, assuming $s \red^k_{m} \lambda t$, we can show $\SpaceBS{s}{m}{\lambda t}$ by induction on $k$. In the base case, $\lambda t \red^0_{m} \lambda t$ implies $m = \size{\lambda t}$ implies $\SpaceBS{\lambda t}{m}{\lambda t}$.

  In the inductive case, $s \red^{\natS k}_{m} \lambda t$ implies a decomposition $s \red s' \red^k_{m'} \lambda t$ for some $s',m'$ with $m=\max(\size s,m')$. Then the inductive hypothesis for $k$ is $\SpaceBS{s'}{m'}{\lambda t}$, which together with \Cref{step-SpaceBS} implies $\SpaceBS{s}{m}{\lambda t}$.
\end{mproof}

Note that especially, if $\Space{s}=m$, then $\SpaceBS{s}{m}{t}$ for some $t$.

\section{Technical Definitions and Lemmas}

\includecollection{remarks}

\section{Proofs}

\includecollection{proofs}

\section{RAM machines can consume more space than time}\label{sec:sveb}

Turing machines can not consume more space than time, since it costs a time unit to allocate a new space unit.
For RAM machines, this is different, as analysed by Slot and van Emde Boas~\cite{Slot:1984:TVC:800057.808705}.

The \emph{time consumption} of a RAM computation is the number of steps; we denote it by $\mathcal{T}$. The \emph{space consumption} is given by
\begin{align*}
~&\mathcal{S}_\mathsf{b} = \sum_{i=0}^{m} \mathsf{size}_\mathsf{b}(i,\mathsf{max}(i))\,,
\end{align*}
where $m$ is the index of the highest address for which a register was accessed and $\mathsf{max}(i)$ is the maximal content of register $\mathcal{R}[i]$ during the computation. Furthermore,
\begin{align*}
\mathsf{size}_\mathsf{b}(i,x) &:= \begin{cases} 0 &\text{if } \mathcal{R}[i] \text{ is unused}\\\log(x)+\log(i) &\text{otherwise} \end{cases}
\end{align*}
Now, intuitively, $\mathcal{S}_\mathsf{b}$ is the sum of the maximum sizes of contents of used registers and the sizes of the addresses required to access those registers. It is known that, using $\mathcal{S}_\mathsf{b}$ as space measure and $\mathcal{T}$ as time measure, RAM machines and Turing machines can simulate each other with a constant overhead in space and a polynomial overhead in time (see e.g. Section~1 in~\cite{Slot:1984:TVC:800057.808705}). Next we consider the following RAM program $P$:
$~$\\

\begin{tabular}{l}
\textbf{Input:} $x$\\
$a \leftarrow 1$;\\
\texttt{for} $i=1$ \texttt{to} $|x|$ \texttt{do} \quad \text{(where $|x|$ is the length of $x$ in binary)}\\
$~~~a \leftarrow a+a$;\\
$~~~\mathcal{R}[a] \leftarrow 1$;\\
\texttt{od};\\
\textbf{Output:} $1$
\end{tabular}
$~$\\
$~$\\
Observe that $P$ is similar to the size-exploding term $s$ we have seen in the introduction. The time consumption of $P$ is given by the following function; recall that the size of an input is given by its length.
\[ n\mapsto 2n+2 \in \bigO{n} \,.\]
However, the space consumption is given by
\begin{align*}
n \mapsto \sum_{i=0}^{2^n} \mathsf{size}_\mathsf{b}(i,\mathsf{max}(i)) = \sum_{i=0}^n \mathsf{size}_\mathsf{b}(2^i,\mathsf{max}(2^i)) = \sum_{i=0}^n \log(2^i)+1 \in \Omega(n^2)\,.
\end{align*}
Thus there are RAM machines that consume asymptotically more space than time, despite being a sequential model. 

\end{document}